\RequirePackage{fix-cm}
\documentclass[smallextended]{svjour3}       
\smartqed  
\usepackage{graphicx}
\usepackage{amssymb}
\usepackage{xcolor}
\usepackage{graphicx}
\usepackage{amsmath}
\usepackage{latexsym}
\usepackage{multirow}
\usepackage{algorithm}
\usepackage{algorithmic}
\usepackage{enumerate}
\usepackage{caption}

\newlength{\commentindent}
\DeclareGraphicsExtensions{.tif,.pdf,.jpeg,.png,.jpg,.bmp,.svg,.eps,.EMF}

\begin{document}

\title{Efficient Decoding Schemes for Noisy Non-Adaptive Group Testing when Noise Depends on Number of Items in Test\thanks{A preliminary version of this paper~\cite{bui2016efficiently} was presented at The International Symposium on Information Theory and Its Applications (ISITA) in Monterey, California, USA, October 30 -- November 2, 2016.}
}

\titlerunning{Dilution model for noisy non-adaptive group testing}        

\author{Thach V. BUI \and Tetsuya KOJIMA \and Minoru KURIBAYASHI \and Isao ECHIZEN}


\institute{Thach V. Bui \at
			SOKENDAI (The Graduate University for Advanced Studies), Kanagawa, Japan. \\
			\email{bvthach@nii.ac.jp}           
           \and
			Tetsuya KOJIMA \at
			National Institute of Technology, Tokyo College, Japan. \\
			\email{kojt@tokyo-ct.ac.jp}           
			\and
			Minoru KURIBAYASHI \at
			Graduate School of Natural Science and Technology, Okayama University, Japan. \\
			\email{kminoru@okayama-u.ac.jp}
			\and
			Isao ECHIZEN \at
			National Institute of Informatics, Tokyo, Japan. \\
			\email{iechizen@nii.ac.jp}			
}

\date{Received: date / Accepted: date}

\maketitle

\begin{abstract}
The goal of non-adaptive group testing is to identify at most $d$ defective items from $N$ items, in which a test of a subset of $N$ items is positive if it contains at least one defective item, and negative otherwise. However, in many cases, especially in biological screening, the outcome is unreliable due to biochemical interaction; i.e., \textit{noise.} Consequently, a positive result can change to a negative one (false negative) and vice versa (false positive). In this work, we first consider the dilution effect in which \textit{the degree of noise depends on the number of items in the test}. Two efficient schemes are presented for identifying the defective items in time linearly to the number of tests needed. Experimental results validate our theoretical analysis. Specifically, setting the error precision of 0.001 and $d\leq16$, our proposed algorithms always identify all defective items in less than 7 seconds for $N=2^{33}\approx 9$ billion.
\end{abstract}

\section{Introduction}
\label{sec:intro}
Group testing has received much attention from researchers worldwide because of its promising applications in various fields~\cite{chen2008survey}. To reduce the number of tests needed to identify infected inductees in WWII, Dorfman~\cite{dorfman1943detection} devised a scheme in which blood samples were mixed together, and the resulting pool was screened. If there were no errors in the blood screening, the outcome of the screening was considered positive if and only if (iff) there was at least one infected blood sample in the pool. In this case, it was called \textit{noiseless group testing}. Otherwise, it was called \textit{noisy group testing.} Nowadays, such group testing is used to find a very small number of \textit{defective items}, say $d$, in a huge number of items, say $N$, using as few \textit{tests}, say $t$, as possible at the lowest cost (decoding time) where $d$ is usually much smaller than $N$. ``Item'', ``defective item'', and ``test'' depend on the context.

\subsection{Background}
\label{sub:background}

There are two main problems in group testing: designing $t$ tests and finding $d$ defective items from $N$ items using the $t$ tests. These main problems are classified into two approaches when applying them in group testing. The first one is \textit{adaptive group testing} in which there are several testing stages, and the design of subsequent stages depends on the results of the previous ones. With this approach, the number of tests can reach the theoretically optimal bound, i.e., $t=\Omega(d\log_2{N})$. However, this can take much time if there are many stages. To overcome this problem, \textit{non-adaptive group testing} (NAGT) has been used. With this approach, all test stages are designed in advance and performed at the same time. This means that all stages can be performed simultaneously. This approach is useful with parallel architectures as it saves time for implementation in various applications, such as multiple access communications \cite{wolf1985born}, data streaming~\cite{cormode2005s}, and fingerprinting codes~\cite{laarhoven2015asymptotics}. To exactly identify all $d$ defective items, it needs $t = \Omega(d^2 \log_2{N}/\log_2{d})$ tests. In this paper, the focus is on NAGT unless otherwise stated.

Noiseless NAGT has been well studied while noisy NAGT, for which the outcome for a test is not reliable, has been well studied in~\cite{atia2012boolean,cheraghchi2013noise,ngo2011efficiently}. If the outcome for a test flips from a positive outcome to a negative one, it is called a \textit{false negative}. If the outcome flips from negative to positive, it is called a \textit{false positive}. There are two assumptions on the number of errors in the test outcome. The first is that there are upper bounds on the number of false positives and the number of false negatives~\cite{cheraghchi2013noise,ngo2011efficiently}; i.e. the maximum number of errors is \textbf{known}. The second and more widely used assumption is that the number of errors is \textbf{unknown}~\cite{atia2012boolean,lee2016saffron}. In the second assumption, there are two types of noise. One is independent noise, where noise is given randomly according to the probability distribution independent of the number of items in each test. Another one is dependent noise, where noise is given according to the number of items in each test. In biological screening, the number of items in a test affects the accuracy of the test outcome~\cite{bruno1995efficient,erlich2015biological,lewis2004measurement,kainkaryam2010poolmc}. Despite the presence of this effect, to the best of our knowledge, there has been no scientific work on the dependent noise case.

\subsection{Our contribution}
\label{sub:contri}
We propose a dilution model for noisy NAGT in which noise depends on the number of items in the test. Then an efficient scheme is proposed for identifying one defective item $(d = 1)$ with high probability by using Chernoff bound. For $d \geq 2$, we utilize the scheme for identifying one defective item to deploy a divide and conquer strategy. As a result, all defective times can be identified with high probability. Moreover, the decoding complexity is linearly scaled to the number of tests for both cases $d = 1$ and $d \geq 2$.

A set of tests needed can be viewed as a binary measurement matrix in which each row represents a test and each column represents for a item. An entry at the row $i$ and the column $j$ is 1 iff the item $j$ belong to the test $i$. A matrix can be constructed nonrandomly if each entry of the matrix can be computed in polynomial time of the number of rows without randomness. That means we save space for storing the measurement matrix and can generate any entry of the matrix at will. In this work, our measurement matrix can be constructed nonrandomly.

Experimental results validate our theoretical analysis. Let the error probability of decoding algorithm be 0.001, the probabilities of false positive and false negative for a test be up to 0.2 and 0.1, respectively. When $d = 1$, the number of tests needed is up to 16,000 even when $N = 2^{33} \approx 9$ billion. Since the number of tests is up to 16,000, the decoding time to identify a defective item is up to 6 microseconds in our experiment environment. When $d \leq 16$, the number of tests is up to 2.5 billion and the decoding time is at most 7 seconds when $N$ is up to $N = 2^{33} \approx 9$ billion.

\subsection{Paper organization}
\label{sub:organiza}
Section~\ref{sec:related} presents related works, including noiseless and noisy group testing, and decoding algorithms. Section \ref{sec:setup} presents the existing and our proposed models. Our proposed schemes for identifying one and more than one defective items are presented in Section~\ref{sec:proposed}. The main results are presented in Section \ref{sec:main} and the evaluation of the proposed schemes is presented in Section~\ref{sec:eval}. The key points are summarized in Section \ref{sec:cls}.

\section{Related works}
\label{sec:related}

The literature on group testing is surveyed here to give an overview on it. Further reading can be found elsewhere~\cite{chen2008survey}. The notation $a := b$ means $a$ is defined in another name for $b$, and $a =: b$ means $b$ is defined in another name for $a$. We can model the group testing problem as follows. Given a binary sparse vector $\mathbf{x} = (x_1, x_2, \ldots, x_N)^T \in \{0, 1 \}^{N}$ representing $N$ items, where $x_j = 1$ iff item $j$ is defective and $|\mathbf{x}| := \sum_{j=1}^N x_j \leq d$, our aim is to design $t$ tests such that $\mathbf{x}$ can be reconstructed at low cost. Each test contains a subset of $N$ items. Hence, a test can be considered to be a binary vector $\{ 0, 1\}^N$ that is associated with the indices of the items belonging to that test. More generally, a set of $t$ tests can be viewed as measurement matrix $\mathcal{T} = (t_{ij}) \in \{0, 1 \}^{t \times N}$, in which each row represents a test, and $t_{ij} = 1$ iff the item $j$ belongs to the test $i$. The outcome for a test is positive (`1') or negative (`0'). Since there are $t$ tests, their outcomes can be viewed as binary vector $\mathbf{y} = (y_1, \ldots, y_t)^T \in \{0, 1 \}^{t}$. The procedure to get the outcome vector $\mathbf{y}$ is called \textit{encoding procedure}. The procedure used to identify defective items from the outcome $\mathbf{y}$, i.e., recovering $\mathbf{x}$, is called \textit{decoding procedure.}

\subsection{Disjunct matrix}
\label{sub:disjunctMatrix}
The union of $l \geq 1$ vectors $\mathbf{y}_1, \ldots, \mathbf{y}_l$ where $\mathbf{y}_i = (y_{1i}, y_{2i}, \ldots, y_{ti})^T$ for $i=1, \ldots, l$ and some integer $t \geq 1$, is defined as vector $\mathbf{y} = \vee_{i = 1}^l \mathbf{y}_i := (\vee_{i = 1}^l y_{1i}, \ldots, \vee_{i = 1}^l y_{ti})^T$, where $\vee$ is the OR operator. A vector $\mathbf{u} = (u_1, \ldots, u_t)^T$ is said not to contain vector $\mathbf{v} = (v_1, \ldots, v_t)^T$ iff there exists the index $i$ such that $u_i = 0$ and $v_i = 1$. To \textit{exactly} identify at most $d$ defective items, the union of at most $d$ columns of $\mathcal{T}$ must not contain other columns~\cite{kautz1964nonrandom}, where each column of $\mathcal{T}$ is treated as a vector here. In the other words, for every $d+1$ columns of $\mathcal{T}$, denoted by $j_0, j_1, \ldots, j_d$, with one designated column, e.g., $j_0$, there exists a row, e.g., $i_0$, such that $t_{i_0 j_0} = 1$ and $t_{i_0 j_a} = 0$, where $a \in \{1, \ldots, d\}$. In short, for every $d+1$ columns of $\mathcal{T}$, there exists an $(d+1) \times (d+1)$ identity matrix constructed by placing those columns in an appropriate order. We call such $\mathcal{T}$ a $t \times N$ $d$-disjunct matrix\footnote{We insist that a $d$-disjunct matrix always identifies exactly $d$ defective items, but it is not necessary for a matrix to be $d$-disjunctive in order to identify $d$ defective items.}. Note that a $d$-disjunct matrix is also a $d^\prime$-disjunct matrix for any integer $d^\prime \leq d$.

The weight of a row (column) is the number of 1s in the row (column). And a matrix is nonrandomly constructed if its entries can be computed in polynomial time of the size of its row withour randomness. Then, from the construction of $d$-disjunct matrices in~\cite{kautz1964nonrandom}, we have the following lemma (the proof is in Appendix \ref{app:Thr1}).
\begin{lemma}
\label{lem:ConstDisjunct}
Given integers $1 \leq d < N$, there exists a nonrandomly constructed $h \times N$ $d$-disjunct matrix with the constant row and column weight, where the weight of every row is less than $N/2$ and $h = O(d^2 \log^2{N})$.
\end{lemma}

\subsection{Noiseless and noisy non-adaptive group testing}
Let $\otimes$ be the Boolean operator for vector multiplication in which multiplication is replaced with the AND operator and addition is replaced with the OR operator. Ideally, if no error occurs in the tests, the outcome of $t$ tests is
\begin{equation}
\mathbf{y} = \mathcal{T} \otimes \mathbf{x} := \begin{bmatrix}
\mathcal{T}_{1, *} \otimes \mathbf{x} \\
\vdots \\
\mathcal{T}_{t, *} \otimes \mathbf{x}
\end{bmatrix} := \bigvee_{j=1}^{N} x_j \mathcal{T}_j = \bigvee_{\substack{j=1 \\ x_j = 1}}^{N} \mathcal{T}_j,
\end{equation}
where $\mathcal{T}_{i, *}$ and $\mathcal{T}_j$ are the $i$th row and the $j$th column of $\mathcal{T}$, respectively. $\mathcal{T}_{i, *} \otimes \mathbf{x}$ equals 1 iff the dot product of $\mathcal{T}_{i, *}$ and $\mathbf{x}$, i.e., $\mathcal{T}_{i, *} \cdot \mathbf{x} = \sum_{j = 1}^N t_{ij} x_j$, is larger than 0. We call this \textit{noiseless} group testing. For exact identification of $d$ defective items, the number of tests is at least $O(d\log_2{N})$ for adaptive group testing and $O(d^2 \log_2{N}/\log_2{d})$ for NAGT~\cite{du2000combinatorial}. If errors do occur in some tests, we call that case \textit{noisy} group testing. There are two types of errors: false positive and false negative. Noisy NAGT has gained attraction and become popular due to its connection to compressed sensing~\cite{atia2012boolean,ngo2011efficiently}. In the previous works cited here, the authors treated noise as a random variable independent of measurement matrix $\mathcal{T}$ or $t$ tests. This is not suitable for certain circumstances, so that it is required to develop a new model in which noise \textit{depends on the number of items in each test}.

An example of noiseless and noisy NAGT is illustrated in Fig. \ref{fig:design}. In this example, the defective set is $\mathbb{G} = \{1, N\}$; i.e., the defective items are 1 and $N$. Ideally, the outcome vector is the union of the columns 1 and $N$ of $\mathcal{T}$ (the rows 1 and $N$ of $\mathcal{T}^T$); that is, $\widehat{\mathbf{y}}^T = \mathcal{T}_1^T \vee \mathcal{T}_N^T$. However, because there is noise, the outcome of the test 2 is flipped from \textit{positive} to \textit{negative} (`1' to `0'), and that of the test 4 is flipped from \textit{negative} to \textit{positive} (`0' to `1'). The final outcome observed is $\mathbf{y}$.  The goal is to find the defective set $\mathbb{G}$ from $\mathbf{y}$.

\begin{figure}[t]
\center
\includegraphics[scale=0.4]{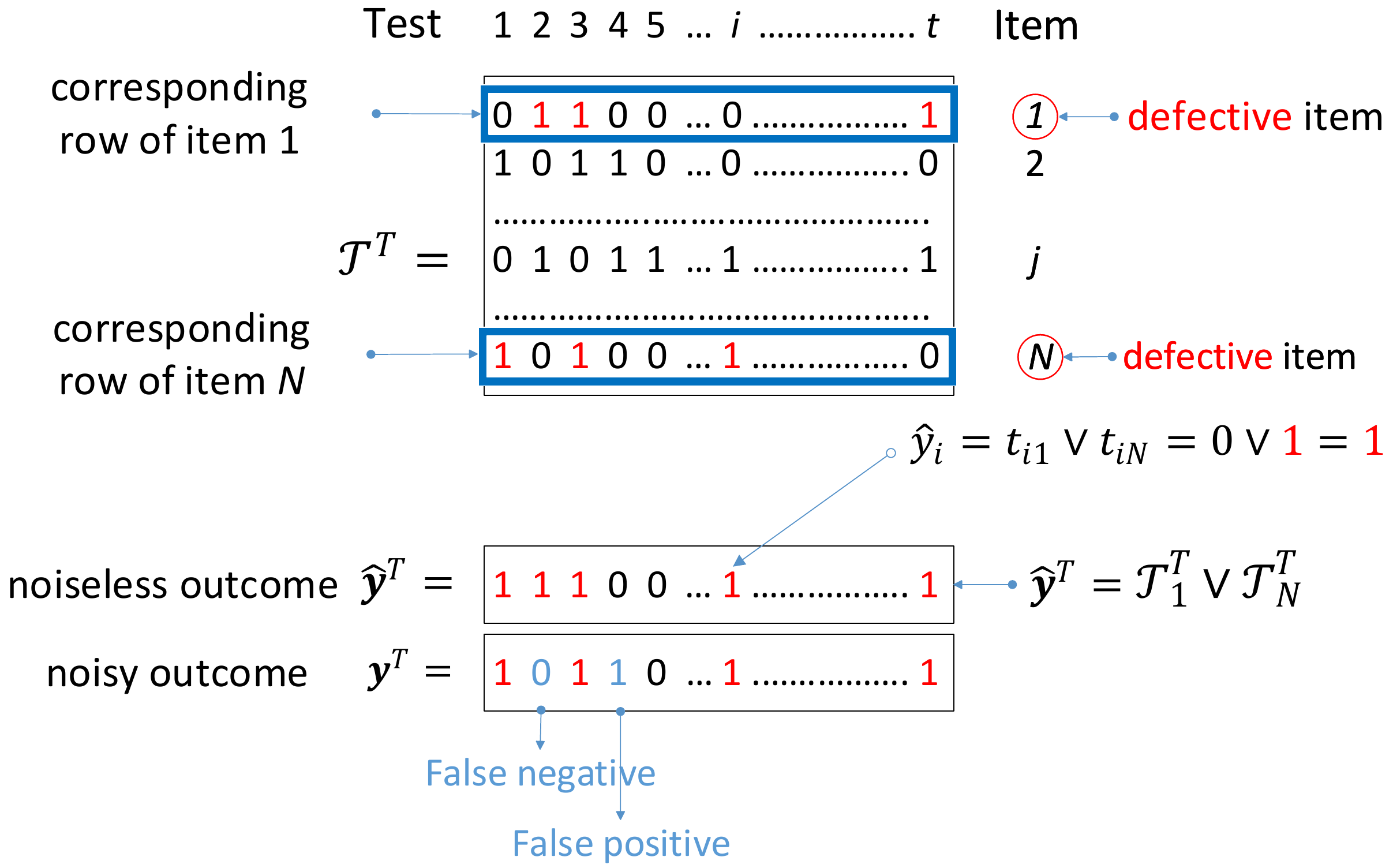}

\caption{A $t \times N$ binary matrix $\mathcal{T} = (t_{ij})$ represents for $t$ tests and $N$ items, where $i = 1,\ldots, t$ and $j = 1, \ldots, N$. Columns and rows represent items and tests, respectively. An entry at the column $j$ and the row $i$ is 1 iff item $j$ belongs to test $i$. In this example, the defective set is $\mathbb{G} = \{1, N \}$; i.e., the defective items are 1 and $N$. Therefore, the outcome of $t$ tests without noise is the union of the columns 1 and $N$, where 0/1 denotes negative and positive outcomes, namely $ \widehat{y}_i = t_{i1} \vee t_{iN}$ for $i = 1, \ldots, t$, where $\widehat{\mathbf{y}} = (\widehat{y}_1, \widehat{y}_2, \ldots, \widehat{y}_t)^T$. In the noisy case, the outcome of the test 2 is flipped from positive (`1') to negative (`0') and that of the test 4 is flipped from negative (`0') to positive (`1'). As a result, the outcome observed is $\mathbf{y}$. Our goal is to find defective items from $\mathbf{y}$.}
\label{fig:design}
\end{figure}

\subsection{Decoding complexity}
The normal decoding complexity in time is $O(tN)$ for noiseless NAGT. If the test outcome is negative, all items contained in the test are non-defective and thus eliminated. The items remaining after eliminating all non-defective items in negative outcomes are defective. Indyk \textit{et al.}~\cite{indyk2010efficiently} and Cheraghchi~\cite{cheraghchi2013noise} made a breakthrough on the decoding complexity problem by presenting a sub-linear algorithm that solves this problem in $\mathrm{poly}(t)$ time. Previous studies have shown that an information-theoretic lower-bound for both the number of tests and decoding complexity~\cite{cai2013grotesque,lee2016saffron} can be almost achieved with $O(d\log_2{d} \cdot \log_2{N})$ tests and implemented in $O(d\log_2{d} \cdot \log_2{N})$ time by using a divide and conquer strategy with high probability. The number of tests and the decoding complexity are much smaller than those with the methods of Indyk \textit{et al.}~\cite{indyk2010efficiently} and Cheraghchi~\cite{cheraghchi2013noise}.

To the best of our knowledge, there has been no algorithm for identifying defective items when noise depends on the number of items in the test.

\subsection{Deterministic and randomized algorithms}
There are four approaches when identifying defective items: to identify all defective items, to identify all defective items with some false positives, to identify a fraction of the defective items (with some false negatives and no false positive), and to identify a fraction of the defective items with some false positives (with/without some false negatives). Algorithms identifying all defective items or a fraction of positive items with probability of 1 are called \textit{deterministic}. The other algorithms with probability less than 1 are called \textit{randomized}.

\section{Problem setup}
\label{sec:setup}

\subsection{Notations}
\label{sub:notation}
We focus on noisy NAGT. For consistency, we use capital calligraphic letters for matrices, non-capital letters for scalars, and bold letters for vectors. All entries of matrices and vectors are binary. Here are some of the notations used:
\begin{itemize}
\item $N, d$: number of items and maximum number of defective items\footnote{For simplicity, we assume that $N$ is to the power of 2.}.
\item $\otimes, \tilde{\otimes}$: operation related to noiseless and noisy NAGT, where $\tilde{\otimes}$ is to be defined later.
\item $\mathcal{T} = (t_{ij})$: $t \times N$ measurement matrix to identify at most $d$ defective items, where integer $t \geq 1$ is the number of tests.
\item $\mathcal{M} = (m_{ij})$: $k \times N$ $1$-disjunct matrix to identify at most one defective item, where integer $k \geq 1$ is the number of tests.
\item $\mathcal{A} = (a_{ij})$: $K \times N$ matrix to identify at most one defective item in noisy NAGT, where integer $K \geq 1$ is the number of tests.
\item $\mathcal{G} = (g_{ij})$: $h \times N$ matrix, where integer $h \geq 1$.
\item $\mathbf{x}, \mathbf{y}, \overline{\mathbf{y}}$:  binary representation of $N$ items, binary representation of the test outcomes, and binary vector of noise.
\item $\mathrm{wt}(.)$: Hamming weight of the input vector, i.e., the number of ones in the vector.
\item $\mathbb{G}$: index set of defective items, e.g., $\mathbb{G} = \{1, 3 \}$ means the items 1 and 3 are defective.
\item $\theta_0, \theta_1$: probability of a false positive and of a false negative.
\item $\mathcal{T}_j, \mathcal{M}_j, \mathcal{G}_{i,*}$: column $j$ of matrix $\mathcal{T}$, column $j$ of matrix $\mathcal{M}$, and row $i$ of matrix $\mathcal{G}$.
\item $\mathrm{diag}(\mathcal{G}_{i, *}) = \mathrm{diag}(g_{i1}, g_{i2}, \ldots, g_{iN})$: diagonal matrix constructed by input vector $\mathcal{G}_{i, *} = (g_{i1}, g_{i2}, \ldots, g_{iN})$.
\item $B(p)$: Bernoulli distribution. A random variable $a$ has the Bernoulli distribution of probability $p$, denoted $a \sim B(p)$ if it takes the value 1 with probability $p$ ($\Pr(a = 1) = p$) and the value 0 with probability $1-p$ ($\Pr(a = 0) = 1 - p$).
\item $\mathrm{e}, \log_2, \ln, \mathrm{exp}(\cdot)$: the base of natural logarithm, the logarithm of base 2, the natural logarithm, and the exponential function.
\item $\lceil a \rceil, \lfloor a \rfloor$: the ceiling and the floor functions of $a$.
\end{itemize}
We note that a $d$-disjunct matrix is a measurement matrix, but a measurement matrix may not be a $d$-disjunct matrix.

\subsection{Existing models}
\label{sub:model}

When working with noisy NAGT, there are two assumptions about the number of errors in test outcome $\mathbf{y}$. The first is that there are upper bounds on the number of false positives and number of false negatives~\cite{cheraghchi2013noise,ngo2011efficiently}; i.e., the number of errors is \textbf{known}. The second and more widely used assumption~\cite{atia2012boolean,lee2016saffron} is that the number of errors is \textbf{unknown}. We only consider the latter assumption here.

When the number of errors is unknown, researchers usually consider two types of noise. One type is independent noise, where noise occurs randomly in accordance with a probability distribution independent of the number of items in each test. The other type is dependent noise, where noise occurs in accordance with the number of items in each test. We use $\mathbf{y} = (y_1, \ldots, y_t)^T$ and $\bar{\mathbf{y}} = (\bar{y}_1, \ldots, \bar{y}_t)^T$ in the sections~\ref{subsub:additive1},~\ref{subsub:additive2},~\ref{subsub:dillu}.

\begin{figure}[ht] 
  \label{errorModels} 
  \begin{minipage}[b]{0.5\linewidth}
    \centering
    \includegraphics[width=0.85\linewidth]{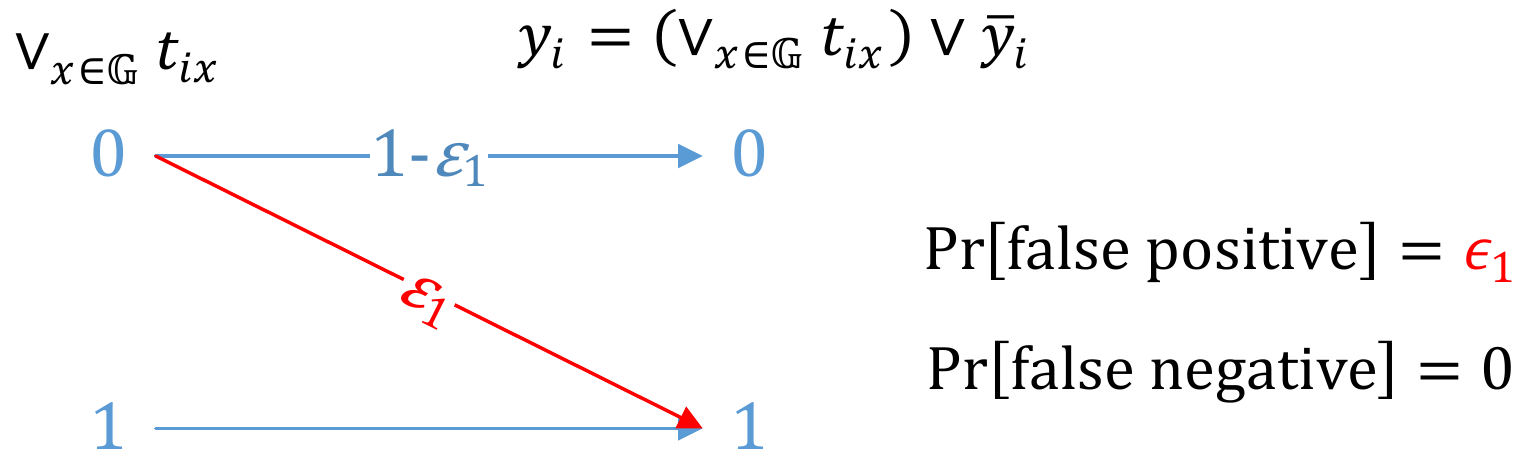}

	\caption{Additive noise type 1 model.}
	\label{additive1}	
	\vspace{4ex}
  \end{minipage}
  \hfill
  \begin{minipage}[b]{0.5\linewidth}
    \centering
    \includegraphics[width=0.85\linewidth]{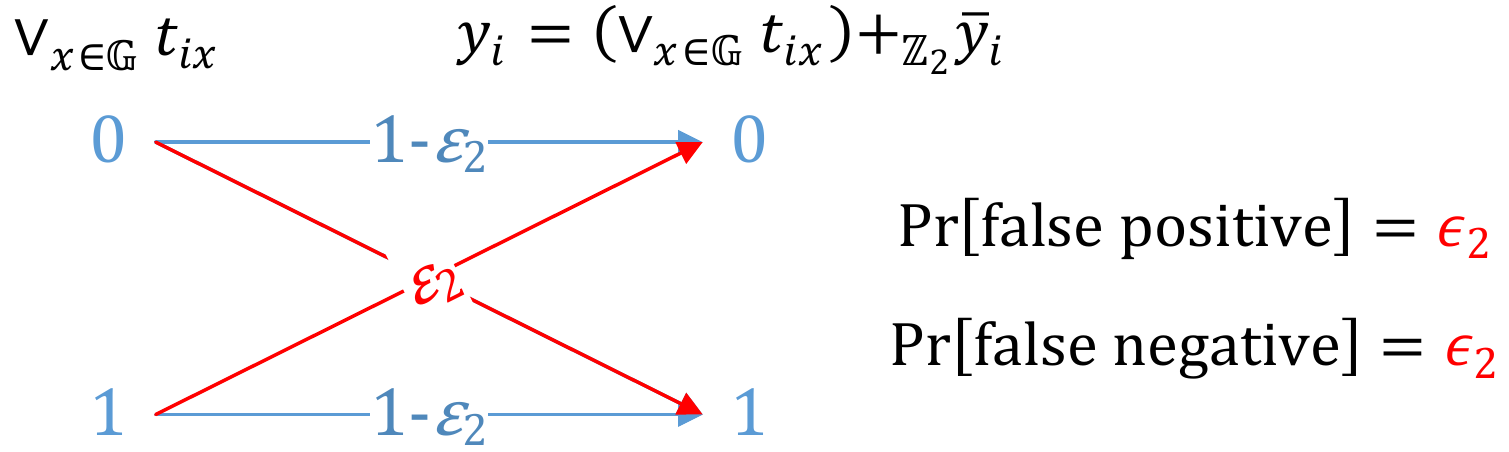}

	\caption{Additive noise type 2 model.}
	\label{additive2}
	\vspace{4ex}
  \end{minipage} 
  
  \begin{minipage}[b]{0.5\linewidth}
    \centering
    \includegraphics[width=0.85\linewidth]{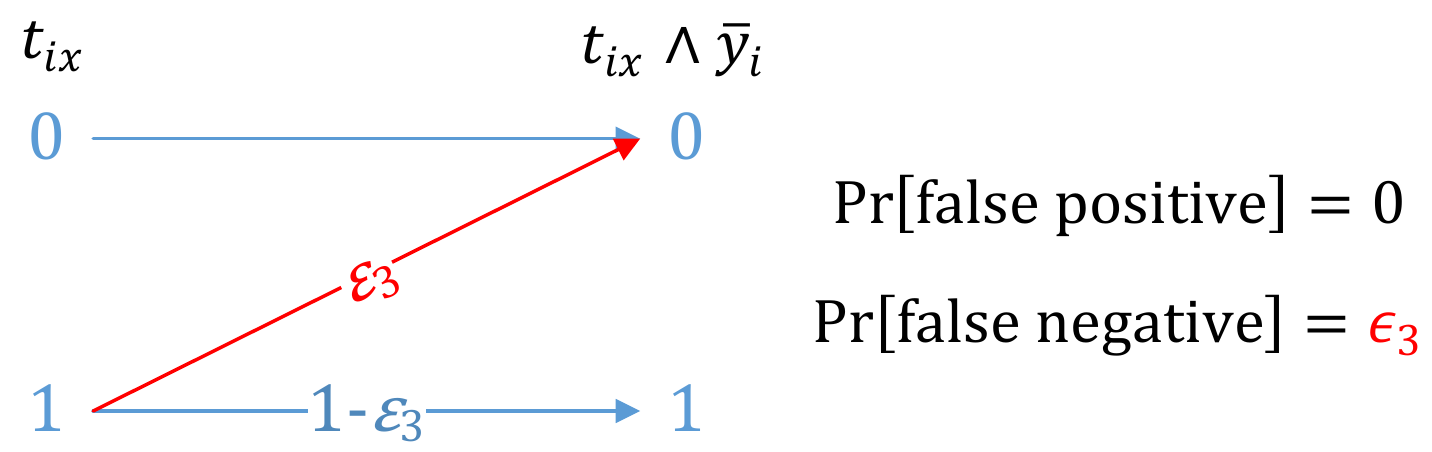}

	\caption{Dilution type 1 model.}
	\label{dilution1}
	\vspace{4ex}
  \end{minipage}
  \hfill
  \begin{minipage}[b]{0.5\linewidth}
    \centering
    \includegraphics[width=0.85\linewidth]{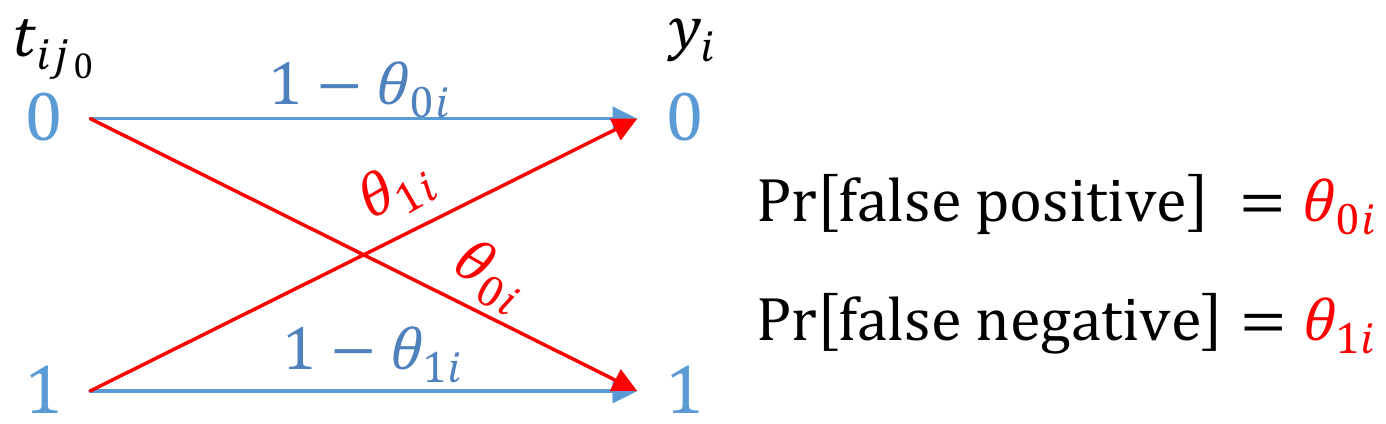}

	\caption{Dilution type 2 model.}
	\label{dilution2}
	\vspace{4ex}
  \end{minipage}
\end{figure}

\subsubsection{Additive noise type 1 model}
\label{subsub:additive1}
Atia and Saligrama~\cite{atia2012boolean} proposed a model in which noise occurs in accordance with a Bernoulli distribution. The observed outcome is $\mathbf{y} = \left(\vee_{x \in \mathbb{G}}\mathcal{T}_x \right) \vee \bar{\mathbf{y}}$, where $\bar{y}_i \sim B(\epsilon_1)$ and $y_i = (\vee_{x \in \mathbb{G}} t_{ix}) \vee \bar{y}_i$ (Fig.~\ref{additive1}) for $i=1, \ldots, t$, and $\epsilon_1 > 0$. In this model, a $t \times N$ measurement matrix $\mathcal{T}$ can be designed with $t = O\left( \frac{d^2 \log{N}}{1 - \epsilon_1} \right)$. Moreover, there is no decoding algorithm associated with this matrix.

\subsubsection{Additive noise type 2 model}
\label{subsub:additive2}
Lee \textit{et al.}~\cite{lee2016saffron} also proposed a model in which noise occurs in accordance with a Bernoulli distribution. However, the error model is different: the observed outcome is $\mathbf{y} = \left( \vee_{x \in \mathbb{G}}\mathcal{T}_x \right) +_{\mathbb{Z}_2} \bar{\mathbf{y}}$, where $\bar{y}_i \sim B(\epsilon)$, $+_{\mathbb{Z}_2}$ is the operation over binary field $\mathbb{Z}_2$ ($0 + 0 = 0, 0 + 1 = 1, 1 + 0 = 1, 1 + 1 = 0$), and $y_i = (\vee_{x \in \mathbb{G}} t_{ix}) +_{\mathbb{Z}_2} \bar{y}_i$ (Fig.~\ref{additive2}) for $i=1, \ldots, t$. In this model, at most $d$ defective items can be identified with probability at least $1 - \delta$ by using a $t \times N$ measurement matrix $\mathcal{T}$ with $t = O \left( \frac{c(\epsilon_2) d \log_2{N}}{1 - H(\epsilon_2) - \delta} \right)$ in time $O(t)$, where $c(\epsilon_2)$ is a value depending on $\epsilon_2 > 0$, $H(\epsilon_2) = \epsilon_2 \log_2{\frac{1}{\epsilon_2}} + (1 - \epsilon_2) \log_2{\frac{1}{1 - \epsilon_2}}$, and $\delta > 0$. In addition, each entry of $\mathcal{T}$ is generated randomly.

\subsubsection{Dilution type 1 model}
\label{subsub:dillu}
The third model is the \textit{dilution model}~\cite{atia2012boolean}, in which $\mathbf{y} = \vee_{x \in \mathbb{G}} (\mathcal{T}_x \wedge \bar{\mathbf{y}})$, where $\bar{y}_i \sim B(1-\epsilon_3)$ and $y_i = \vee_{x \in \mathbb{G}} (t_{ix} \wedge \bar{y}_i)$ for $i=1, \ldots, t$. Each entry of $\mathcal{T}_x \wedge \bar{\mathbf{y}}$ has the distribution shown as Fig.~\ref{dilution1}, i.e., the probability of false positive is 0 and of false negative is $\epsilon_3$. In this model, a $t \times N$ measurement matrix $\mathcal{T}$ can be designed with $t = O \left( \frac{d^2 \log_2{N}}{(1-\epsilon_3)^2} \right)$. Moreover, there is no decoding algorithm associated with this matrix.

\subsection{Proposed model: dilution type 2 model}
\label{sub:new}
In some circumstances, especially in biological screening, the three noise models above may not be suitable. The noise in biological screening depends on the number of items in the test, and there have been no reports on this noise model. Assume that there is only \textbf{one defective item} $j_0$ in $N$ items. Then the binary representation vector of $N$ items $\mathbf{x} := (x_1, \ldots, x_N)^T$ has the only non-zero entry $x_{j_0} = 1$. We define the outcome of test $i$ in noisy NAGT under dilution type 2 model is as follows:
\begin{equation}
\begin{cases}
\Pr(\mathrm{outcome} = 1 | t_{ij_0} = 0) := \Pr( \mathcal{T}_{i, *} \tilde{\otimes} \mathbf{x} = 1 | t_{ij_0} = 0) = \theta_{0i} \leq \frac{1}{2} \\
\Pr(\mathrm{outcome} = 0 | t_{ij_0} = 1) := \Pr( \mathcal{T}_{i, *} \tilde{\otimes} \mathbf{x} = 0 | t_{ij_0} = 1) = \theta_{1i} \leq \frac{1}{2}
\end{cases}
\label{eqn:prob1Defec}
\end{equation}
where $\tilde{\otimes}$ is the operation for noisy NAGT\footnote{It is reasonable to assume the probabilities of false positive and false negative are at most $\frac{1}{2}$. Otherwise, our test designs must not be used because of its high unreliability.}. When there is at most one defective item in \textit{a test}, the first and second equations indicate the probability of false positive and false negative, respectively. Formally, we define the operation $\tilde{\otimes}$ in case $|\mathbf{x}| \leq 1$ as follows:
\begin{eqnarray}
\Pr( \mathcal{T}_{i, *} \tilde{\otimes} \mathbf{x} = 1) = \begin{cases}
\theta_{0i} \leq \frac{1}{2} \qquad \hfill \mbox{ if } t_{ij_0} = 0 \\
1 - \theta_{0i} \geq \frac{1}{2} \qquad \hfill \mbox{ if } t_{ij_0} = 1
\end{cases} \\
\Pr( \mathcal{T}_{i, *} \tilde{\otimes} \mathbf{x} = 0) = \begin{cases}
\theta_{1i} \leq \frac{1}{2} \qquad \hfill \mbox{ if } t_{ij_0} = 1 \\
1 - \theta_{1i} \geq \frac{1}{2} \qquad \hfill \mbox{ if } t_{ij_0} = 0
\end{cases}
\end{eqnarray}

We note that there is no \textit{explicit function} for the dependency on the number of items in the test and on $\theta_{0i}$ and $\theta_{1i}$. This model has an \textbf{asymmetric} noise channel in which the probability of false positive $\theta_{0i}$ and the probability of false negative $\theta_{1i}$ are free from having any explicit relationship such as $\theta_{0i} = \theta_{1i}$, as illustrated in Fig.~\ref{dilution2}, where $i=1, \ldots, t$. Eq.~\eqref{eqn:prob1Defec} also shows that the number of defective items in a test is at most one in this model.

Assume that $\mathrm{wt}(\mathcal{T}_{i, *}) = m_i$ for $i = 1, \ldots, t$. From Eq. \eqref{eqn:prob1Defec}, the probability that the outcome is correct is:
\begin{eqnarray}
p_{0i} &=& \Pr(\mbox{outcome is correct}) \notag \\
&=& \Pr(\mbox{outcome} = 0 | t_{ij_0} = 0) \Pr(t_{ij_0} = 0) \notag \\
&+& \Pr(\mbox{outcome} = 1 | t_{ij_0} = 1) \Pr(t_{ij_0} = 1) \\
&=& (1 - \theta_{0i}) \times \left( 1 - \frac{m_i}{N} \right) + (1 - \theta_{1i}) \times \frac{m_i}{N} \label{eqn:probWeight}
\end{eqnarray}
\noindent
Eq. \eqref{eqn:probWeight} is derived because $\Pr(t_{ij_0} = 1) = \frac{\mathrm{wt}(\mathcal{T}_{i, *})}{N} = \frac{m_i}{N}$.

We are now going to evaluate the expectation of the number of correct outcomes. Let denote the event that the outcome of test $i$ is correct as $X_i$. $X_i$ takes 1 if the outcome of test $i$ is correct, and 0 otherwise. Then,
\begin{eqnarray}
\mathbb{E}(X_i) &=& 1 \times \Pr(X_i = 1) + 0 \times\Pr(X_i = 0) \\
&=& \Pr(\mbox{outcome is correct}) = p_{0i} \\
&=& (1 - \theta_{0i}) \times \left( 1 - \frac{m_i}{N} \right) + (1 - \theta_{1i}) \times \frac{m_i}{N} \label{eqn:expectWeight}
\end{eqnarray}

Assume that there are $c$ independent tests $1, 2, \ldots, c$. Let $X := \sum_{i = 1}^c X_i$ denote their sum and let $\mu := \mathbb{E}(X)$ denote the expectation of $X$. The variable $X$ indicates \textbf{how many outcomes are correct} in $c$ tests. Then, $\mu = \mathbb{E}(X) = \mathbb{E} \left( \sum_{i = 1}^c X_i \right) = \sum_{i = 1}^c  \mathbb{E} (X_i) = \sum_{i = 1}^c p_{0i}$. Assume that $ \left( \frac{1}{2} + \xi \right) c \leq (1 - \lambda) \mu$ for some $\xi > 0$ and any $\lambda > 0$. Using Chernoff's bound, we have:
\begin{eqnarray}
\Pr \left( X \leq \left( \frac{1}{2} + \xi \right)c \right) \leq \Pr \left(X \leq \left(1 - \lambda \right) \mu \right) \leq \mathrm{exp} \left(- \frac{\lambda^2 \mu}{2} \right). \label{ChernoffX}
\end{eqnarray}
Eq.~\eqref{ChernoffX} implies that the probability where the number of correct outcome is not dominant $\left( \Pr \left( X \leq \frac{1}{2}c \right) \right)$ is at most $\mathrm{exp} \left(- \frac{\lambda^2 \mu}{2} \right)$.

Assume that we have a test under the dilution type 2 model. We would like to know the outcome of the test in noiseless setting, i.e., $\mathcal{T}_{i, *} \otimes \mathbf{x}$. One can design tests as follows. Repeat $c$ independent trials for the test. Then, we count how many outcomes are 1 and how many outcomes are 0. Then if the number of positive (negative) outcome is larger than $\frac{1}{2} c$, we claim that the outcome of the test in noiseless setting is positive (negative). Because of Eq.~\eqref{ChernoffX}, our claim is wrong with probability at most $\mathrm{exp} \left(- \frac{\lambda^2 \mu}{2} \right)$. The parameters $\xi$ and $\lambda$ will be specified in the next section.

\subsubsection{The case of equal number of items in each test}
\label{subsub:Dilu2Equal}
We instantiate the model in Fig.~\ref{dilution2} in the case when there is at most one defective item ($d = 1$). Assume that the number of items in \textit{every test} is equal to $N/2$, i.e., $m = m_1 = \ldots = m_t = \frac{N}{2}$. Therefore, the probability of a false positive and of a false negative are $ 0 \leq \theta_{0} = \theta_{01} = \ldots = \theta_{0t} \leq \frac{1}{2}$ and $ 0 \leq \theta_{1} = \theta_{11} = \ldots = \theta_{1t} \leq \frac{1}{2}$, respectively.

Using Eq. \eqref{eqn:probWeight} and $\frac{m}{N} = \frac{1}{2}$, the probability where an outcome is correct is:
\begin{eqnarray}
p_0 = p_{01} = \ldots = p_{0t} = 1 - \frac{1}{2} (\theta_0 + \theta_1).
\label{p0}
\end{eqnarray}

Hence, $\mathbb{E}(X_1) = \ldots = \mathbb{E}(X_t) = p_0 $ in Eq. \eqref{eqn:expectWeight} and $\mu_0 = \sum_{i = 1}^c  \mathbb{E} (X_i) = p_0 c$. The condition $ \left( \frac{1}{2} + \xi \right) c \leq (1 - \lambda) \mu$ becomes:
\begin{alignat}{3}
&&  \left( \frac{1}{2} + \xi \right) c &\leq (1 - \lambda) \mu_0 = (1 - \lambda)p_0 c \\
&\Longleftrightarrow& p_0 = 1 - \frac{1}{2} (\theta_0 + \theta_1)  &\geq \frac{1/2 + \xi}{1 - \lambda}  \\
&\Longleftrightarrow& \theta_0 + \theta_1 & \leq 2 \left( 1 - \frac{1/2 + \xi}{1 - \lambda} \right). \label{conditionTheta}
\end{alignat}

Thus, Eq. \eqref{ChernoffX} becomes
\begin{eqnarray}
\Pr \left( X \leq \left( \frac{1}{2} + \xi \right)c \right) \leq \Pr \left(X \leq \left(1 - \lambda \right) \mu_0 \right) \leq \mathrm{exp} \left( -\frac{\lambda^2 \mu_0}{2}  \right). \label{Chernoff2Equal}
\end{eqnarray}

If we carefully choose $\lambda$ and $\xi$, the right term of Eq.~\eqref{conditionTheta} is almost 1. Then, that inequality always holds.

\subsubsection{The case of unequal number of items in each test}
\label{subsub:Dilu2Diff}
We instantiate the model in Fig.~\ref{dilution2} when the number of defective items is at most one ($d = 1$), and the number of items in \textbf{every test} does not \textit{exceed} $N/2$. In the latter sections, we will carefully design such tests. Obviously, the assumption when the number of items in every test is equal to $N/2$ in section~\ref{subsub:Dilu2Equal} is a special case of this assumption. 

Naturally, in biological screening~\cite{griffin2000diabetes,lewis2004measurement}, if the number of items is decreased while keeping the same number of defective items, the probability that the outcome of a test is correct increases. Since the number of items in every test in section~\ref{subsub:Dilu2Equal} is $N/2$ and the number of items in this model is at most $N/2$, one implies $p_{0i} \geq p_0$. Thus, $\mu = \sum_{i = 1}^c  \mathbb{E} (X_i) \geq \sum_{i = 1}^c  p_0 = \mu_0$. Because $\left( \frac{1}{2} + \xi \right) c \leq (1 - \lambda) \mu_0$ (condition in Eq.~\eqref{conditionTheta}) and $\mu_0 \leq \mu$, $\left( \frac{1}{2} + \xi \right) c \leq (1 - \lambda) \mu$. That means Eq.~\eqref{conditionTheta} always holds in this case. Therefore, Eq. \eqref{ChernoffX} becomes
\begin{eqnarray}
\Pr \left( X \leq \left( \frac{1}{2} + \xi \right)c \right) &\leq& \Pr \left(X \leq \left(1 - \lambda \right) \mu \right) \leq \mathrm{exp} \left( -\frac{\lambda^2 \mu}{2}  \right) \\
&\leq& \mathrm{exp} \left( -\frac{\lambda^2 \mu_0}{2}  \right) \label{Chernoff2NotEqual}
\end{eqnarray}

\subsubsection{On the number of defective items in a test}
\label{subsub:more1InTest}
Since Eq.~\eqref{ChernoffX}, Eq.~\eqref{Chernoff2Equal}, and Eq.~\eqref{Chernoff2NotEqual} are only applicable when the number of defective items in a test is up to one, we study the case when the number of defective items in a test is more than one here. Note that the model in Fig.~\ref{dilution2} is not applicable in this case. Let $i$ and $i^\prime$ be the tests having \textit{at most} one defective item and \textit{more than} one defective item, respectively. They also share the same number of items in the tests where the number of items does not exceed $N/2$. Let $p_{0i}^\prime$ be the probability that the outcomes of test $i^\prime$ is correct. From biological screening~\cite{griffin2000diabetes,lewis2004measurement}, for the same number of items in a test, the probability where the outcome is correct increases as the number of defective items increases. Then $p_{0i}^\prime \geq p_{0i} \geq p_0$.

Let denote the event that the outcome of test $i^\prime$ is correct as $X_i^\prime$. $X_i^\prime$ takes 1 if the outcome of test $i$ is correct and 0 otherwise. Then, 
\begin{eqnarray}
\mathbb{E}(X_i^\prime) &=& 1 \times \Pr(X_i^\prime = 1) + 0 \times\Pr(X_i^\prime = 0) \\
&=& \Pr(\mbox{outcome is correct}) = p_{0i}^\prime.
\end{eqnarray}

Assume that test $i^\prime$ is repeated for $c$ independent trials. Let $X_i^\prime$ be a random variable that takes 1 if the outcome of trial $s$ is correct, and 0 otherwise, where $s = 1, \ldots, c$. Let $X^\prime := \sum_{s = 1}^c X_s^\prime$ denote their sum and let $\mu^\prime := \mathbb{E}(X^\prime)$ denote the expectation of $X^\prime$. Again, the variable $X^\prime$ indicates \textbf{how many outcomes are correct} in $c$ trials\footnote{$X^\prime$ is used here to distinguish with notation $X$ in Eq.~\eqref{ChernoffX}. However, they share the same meaning, which indicates how many outcomes are correct.}. Then, $\mu^\prime = \mathbb{E}(X^\prime) = \mathbb{E} \left( \sum_{s = 1}^c X_s^\prime \right) = \sum_{s = 1}^c  \mathbb{E} (X_i^\prime) = \sum_{s = 1}^c p_{0i}^\prime = p_{0i}^\prime c \geq p_0 c = \mu_0$. For some $\xi > 0$ and any $\lambda > 0$, because $\left( \frac{1}{2} + \xi \right) c \leq (1 - \lambda) \mu_0$ (condition in Eq.~\eqref{conditionTheta}), $\left( \frac{1}{2} + \xi \right) c \leq (1 - \lambda) \mu_0 \leq (1 - \lambda) \mu^\prime$. That means Eq.~\eqref{conditionTheta} always holds in this case. Using Chernoff's bound, we have:
\begin{eqnarray}
\Pr \left( X^\prime \leq \left( \frac{1}{2} + \xi \right)c \right) &\leq& \Pr \left(X^\prime \leq \left(1 - \lambda \right) \mu^\prime \right) \leq \mathrm{exp} \left(- \frac{\lambda^2 \mu^\prime}{2} \right) \\
&\leq& \mathrm{exp} \left( -\frac{\lambda^2 \mu_0}{2}  \right). \label{ChernoffGeneral}
\end{eqnarray}

\subsubsection{Summary}
\label{subsub:sumDilution}
\textbf{Model:} Let $\theta_0$ and $\theta_1$ be the probabilities of false positive and false negative in a test in which there is at most one defective item among $N$ items, and the number of items in a test is $N/2$. For some $\xi > 0$ and any $\delta > 0$, assume that $\theta_0 + \theta_1  \leq 2 \left( 1 - \frac{1/2 + \xi}{1 - \lambda} \right)$. Moreover, any test has at most $N/2$ items and there is no restriction on the number of defective items in a test. We also assume that the probability where an outcome is correct is at least $p_0$ (as analyzed in section~\ref{subsub:Dilu2Equal},~\ref{subsub:Dilu2Diff} and~\ref{subsub:more1InTest}).

When the number of items in every tests equals to $N/2$ and there is only one defective item in $N$ item, this model reduces to the model in section~\ref{subsub:Dilu2Equal}. Note that there is \textbf{no explicit} function for operation $\tilde{\otimes}$ when $|\mathbf{x}| \geq 2$. The only general information induced from $|\mathbf{x}|$ is that the probability where the outcome of test $i$, i.e., $\mathcal{T}_{i, *} \tilde{\otimes} \mathbf{x}$, is correct is at least $p_0$.

In this model, the outcome of $t$ tests using a $t \times N$ measurement matrix $\mathcal{T}$ can be formulated using operator $\tilde{\otimes}$ as follows:
\begin{eqnarray}
\mathbf{y} = \mathcal{T} \tilde{\otimes} \mathbf{x} = \begin{bmatrix}
\mathcal{T}_{1, *} \tilde{\otimes} \mathbf{x} \\
\vdots \\
\mathcal{T}_{t, *} \title{\otimes} \mathbf{x}
\end{bmatrix} \in \{0, 1 \}^t.
\end{eqnarray}

Suppose that a test has at most $N/2$ items and is repeated for $c$ independent trials. Let $Y$ be the number of the correct outcomes in $c$ trials. From Eq.~\eqref{Chernoff2Equal}, Eq.~\eqref{Chernoff2NotEqual}, and Eq.~\eqref{ChernoffGeneral}, the following inequality holds without distinction on the number of defective items in the test:
\begin{eqnarray}
\Pr \left( Y \leq \frac{1}{2} c \right) \leq \Pr \left( Y \leq \left( \frac{1}{2} + \xi \right)c \right) \leq \mathrm{exp} \left( -\frac{\lambda^2 \mu_0}{2}  \right), \label{HerculeanChernoff}
\end{eqnarray}
where $\mu_0 = p_0 c = \left(1 - \frac{1}{2}(\theta_0 + \theta_1) \right)c$.

Then, the simple rule to identify the outcome of the test in noiseless setting: if the number of positive (negative) outcome is larger than $\frac{1}{2} c$, the outcome of the test in noiseless setting is positive (negative). Because of Eq.~\eqref{HerculeanChernoff}, our claim is wrong with probability at most $\mathrm{exp} \left(- \frac{\lambda^2 \mu_0}{2} \right)$.

\section{Proposed scheme}
\label{sec:proposed}

\subsection{Overview of divide and conquer strategy}
\label{sub:divide}
Our interest in decoding one defective item was sparked by the schemes proposed by Lee et al.~\cite{lee2016saffron} and Cai et al.~\cite{cai2013grotesque}. These schemes make defective items easier to decode using a divide and conquer strategy. If there are more than one defective item in a test, the test outcome is always positive. Then, it is difficult to identify them because it is not sure that how many defective items are in the test. Therefore, if there is only one defective item in a test, it is easier to find that item. This means there should be at least $d$ tests among $t$ tests, each contains only one defective item and all defective items belongs to those $d$ tests. However, these schemes are \textbf{incompatible} to the dilution type 2 model. Therefore, we propose schemes to solve this problem.

\subsection{Decoding a defective item}
\label{sub:decoding1}

\subsubsection{Overview}
\label{subsub:dec1Overview}

Here we consider the dilution type 2 model in section~\ref{subsub:sumDilution} in which the number of defective items in a test is arbitrary, but the number of items in the test is up to $N/2$. Let $\mathcal{B}$ be a $k \times N$ measurement matrix and $\mathbf{x} = (x_1, \ldots, x_N)^T \in \{0, 1 \}^N$ be the binary representation of $N$ items in which $x_j = 1$ iff item $j$ is defective. Assume that $\mathbf{y}^\prime  \in \{0, 1 \}^k$ is the outcome vector of $k$ tests, i.e., $\mathbf{y}^\prime = \mathcal{B} \otimes \mathbf{x}$. $\mathcal{B}$ is designed as follows:
\begin{enumerate}[(i)]
\item If $\mathbf{y}^\prime$ is a non-zero column of $\mathcal{B}$, the index of the column is always identified.
\item If $\mathbf{y}^\prime$ is a zero column or the union of at least two non-zero columns of $\mathcal{B}$, there is no defective item identified.
\item The Hamming weight of any row of $\mathcal{B}$ is at most $N/2$.
\end{enumerate}
Item $j$ is in vector $\mathbf{r}$ if its corresponding column in a matrix, e.g., $\mathbf{j}$, belongs to $\mathbf{r}$, i.e., $\mathbf{j} \vee \mathbf{r} = \mathbf{r}$. Condition (i) ensures that if there is only one defective item among $N$ items in the outcome vector, it is always identified. Condition (ii) ensures that there is no defective item identified in the case where there is none or more than one defective item in the outcome vector. Condition (iii) ensures that the dilution type 2 model in section~\ref{subsub:sumDilution} holds.

A bigger $K \times N$ matrix $\mathcal{A}$, which is generated from $\mathcal{B}$, is used to deploy tests. To comply the condition in the dilution type 2 model in section~\ref{subsub:sumDilution}, the weights of any rows in $\mathcal{B}$ and $\mathcal{A}$ are at most $N/2$. To avoid ambiguous notations and misunderstanding, $\mathcal{A}$ is denoted for the measurement matrix instead of $\mathcal{T}$ as usual because it will be used in case $d \geq 2$ in the latter section~\ref{sub:decodingD}.

The basic idea of our scheme can be summarized as follows. Given a $k \times N$ matrix $\mathcal{B} = \begin{bmatrix}
\mathcal{B}_{1, *} \\
\vdots \\
\mathcal{B}_{k, *}
\end{bmatrix}$, a $K \times N$ measurement matrix $\mathcal{A}$, where $K = ck$ for some $c > 0$, is created as:
\begin{eqnarray}
\mathcal{A} := \begin{bmatrix}
\mathcal{B}_{1, *} \\
\vdots \\
\mathcal{B}_{1, *} \\
\hline
\vdots \\
\hline
\mathcal{B}_{k, *} \\
\vdots \\
\mathcal{B}_{k, *}
\end{bmatrix}
\begin{array}{l}
\left.
\begin{array}{l}
\ \\
\ \\
\ \\
\end{array}
\right\} \ c \mbox{ times} \\
\quad \ \vdots \\
\left.
\begin{array}{l}
\ \\
\ \\
\ \\
\end{array}
\right\} \ c \mbox{ times}
\end{array}
\end{eqnarray}
\noindent

Our goal is to derive $\mathbf{y}^\prime = \mathcal{B} \otimes \mathbf{x} = \begin{bmatrix}
y_1^\prime \\
\ldots \\
y_k^\prime
\end{bmatrix} \in \{ 0, 1 \}^k$ from
\begin{eqnarray}
\mathbf{y}_* := \begin{bmatrix}
\\
\mathbf{y}_1 \\
\\
\hline 
\vdots \\
\hline 
\\
\mathbf{y}_k \\
\\
\end{bmatrix}:= \begin{bmatrix}
y_{1}^1 \\
\vdots \\
y_{1}^c \\
\hline 
\vdots \\
\hline 
y_{k}^1 \\
\vdots \\
y_{k}^c
\end{bmatrix} =: \begin{bmatrix}
\mathcal{B}_{1, *} \tilde{\otimes} \mathbf{x} \\
\vdots \\
\mathcal{B}_{1, *} \tilde{\otimes} \mathbf{x} \\
\hline
\vdots \\
\hline
\mathcal{B}_{k, *} \tilde{\otimes} \mathbf{x} \\
\vdots \\
\mathcal{B}_{k, *} \tilde{\otimes} \mathbf{x}
\end{bmatrix}
= \mathcal{A} \tilde{\otimes} \mathbf{x} \in \{0, 1 \}^K,
\end{eqnarray}
which is what we observe. Note that $y_i^\prime = \mathcal{B}_{i, *} \otimes \mathbf{x}$ and $\mathbf{y}_i = \begin{bmatrix}
y_i^1 \\
\vdots \\
y_i^c
\end{bmatrix} = \begin{bmatrix}
\mathcal{B}_{i, *} \tilde{\otimes} \mathbf{x} \\
\vdots \\
\mathcal{B}_{i, *} \tilde{\otimes} \mathbf{x}
\end{bmatrix} \in \{0, 1 \}^c$ for $i = 1, \ldots, k$. Then $\mathbf{x}$ can be recovered by using $\mathcal{B}$. On the other hand, $\mathbf{x}$ is efficiently recovered as $\mathcal{B}$ is efficiently decoded. Matrix $\mathcal{B}$ will be addressed in section~\ref{subsub:extSAFFRON}.

The idea can be illustrated here:
\begin{eqnarray}
\mathcal{B} \xrightarrow[\text{ENCODING}]{\text{enlarge}} \mathcal{A} \xrightarrow[\text{ENCODING}]{\tilde{\otimes} \mathbf{x}} \mathbf{y}_* = \mathcal{A} \tilde{\otimes} \mathbf{x} \xrightarrow[\text{DECODING}]{} &&\mathbf{y}^\prime = \mathcal{B} \otimes \mathbf{x} \notag \\
&&\xrightarrow[\text{DECODING}]{\mathcal{B}} \begin{tabular}{@{}c@{}} (No) defective \\ item \end{tabular} \notag
\end{eqnarray}

\subsubsection{The SAFFRON scheme}
\label{subsub:SAFFRON}
\textit{Encoding procedure:} Lee et al.~\cite{lee2016saffron} propose the following $k \times N$ measurement matrix:
\begin{equation}
\label{eqn:bin}
\mathcal{M} := \begin{bmatrix}
\mathbf{b}_1 & \mathbf{b}_2 \ldots \mathbf{b}_N \\
\overline{\mathbf{b}}_1 & \overline{\mathbf{b}}_2 \ldots \overline{\mathbf{b}}_N
\end{bmatrix} =
\begin{bmatrix}
\mathcal{M}_1 \ldots \mathcal{M}_N
\end{bmatrix}
\end{equation}
where $k = 2\log_2{N}$, $\mathbf{b}_j$ is the $\log_2{N}$-bit binary representation of integer $j-1$, $\overline{\mathbf{b}}_j$ is $\mathbf{b}_j$'s complement, and $\mathcal{M}_j := \begin{bmatrix} \mathbf{b}_j \\ \overline{\mathbf{b}}_j \end{bmatrix}$ for $j = 1,2,\ldots, N$. Remember that $\mathrm{wt}(\mathcal{M}_j) = \log_2{N}$ and $\mathrm{wt}(\mathcal{M}_{i, *}) = N/2$ for $i = 1, \ldots, k$.

For any $\mathbf{x} = (x_1, \ldots, x_N)^T \in \{0, 1 \}^N$, we have:
\begin{equation}
\mathcal{M} \otimes \mathbf{x} = \bigvee_{j = 1}^N x_j\mathcal{M}_j = \bigvee_{\substack{j = 1 \\ x_j = 1}}^N \mathcal{M}_j
\label{SaffronEq}
\end{equation}

\textit{Decoding procedure:} Thanks to the fact that the Hamming weight for each column of $\mathcal{M}$ is $\log_2{N}$, given any vector of size $k = 2\log_2{N}$, if its Hamming weight equals to $k/2 = \log_2{N}$, the index of the defective item is derived from its first half ($\mathbf{b}_j$ for some $j \in \{1, \ldots, N \}$). If a vector is the union of at least 2 columns of $\mathcal{M}$ or zero vector, the Hamming weight of that vector is never equal to $\log_2{N}$ because of Eq.~\eqref{SaffronEq}. Therefore, given an input outcome vector, we can either identify the defective item or there is no single defective item in the outcome vector. The latter case is considered that no defective item in the outcome vector.

For example, let us consider the case $N = 8, k = 2\log_2{N} = 6$. Assume that
\begin{eqnarray}
\label{eqn:M}
\mathcal{M} = 
\begin{bmatrix}
0 & 0 & 0 & 0 & 1 & 1 & 1 & 1 \\
0 & 0 & 1 & 1 & 0 & 0 & 1 & 1 \\
0 & 1 & 0 & 1 & 0 & 1 & 0 & 1 \\
\hline
1 & 1 & 1 & 1 & 0 & 0 & 0 & 0 \\
1 & 1 & 0 & 0 & 1 & 1 & 0 & 0 \\
1 & 0 & 1 & 0 & 1 & 0 & 1 & 0 \\
\end{bmatrix},
\end{eqnarray}
\noindent
and the outcome vectors are $\mathbf{m}_1 = (0, 0, 0, 1, 1, 1)^T$ and $\mathbf{m}_2 = (1, 0, 1, 1, 1, 0)^T$. Since $\mathrm{wt}(\mathbf{m}_1) = 3 = \log_2{N}$, the index of the defective item is 1 because its first half is $(0, 0, 0)$. Similarly, since $\mathrm{wt}(\mathbf{m}_2) = 4 \neq 3 = \log_2{N}$, there is no defective item corresponding to this outcome vector $\mathbf{m}_2$.

\subsubsection{Extension of the SAFFRON scheme}
\label{subsub:extSAFFRON}
Given $\mathbf{u} = (u_1, \ldots, u_L) \in \{0, 1\}^L$ and $\mathbf{v} = (v_1, \ldots, v_L) \in \{ 0, 1\}^L$, we define $\mathbf{u} \wedge \mathbf{v} := \left(u_1 \wedge v_1, \ldots, u_L \wedge v_L \right)$. Assume that $\mathbf{g} = (g_1, \ldots, g_N) \in \{0, 1 \}^N$, then $\mathcal{B}$ is constructed as follows:
\begin{equation}
\label{matrixB}
\mathcal{B} := \mathcal{M} \times \mathrm{diag}(\mathbf{g}) = \begin{bmatrix}
g_{1}\mathcal{M}_1 & \ldots & g_{N}\mathcal{M}_N
\end{bmatrix} = \begin{bmatrix}
\mathcal{M}_{1, *} \wedge \mathbf{g} \\
\vdots \\
\mathcal{M}_{k, *} \wedge \mathbf{g} \\
\end{bmatrix} = \begin{bmatrix}
\mathcal{B}_{1, *} \\
\vdots \\
\mathcal{B}_{k, *}
\end{bmatrix}
\end{equation}
where $\mathrm{diag}(\mathbf{g}) = \mathrm{diag}(g_{1}, \ldots, g_{N})$ is the diagonal matrix constructed by input vector $\mathbf{g}$, and $\mathcal{B}_{i, *} = \mathcal{M}_{i, *} \wedge \mathbf{g}$ for $i = 1, \ldots, k$. It is easy to confirm that $\mathcal{B} = \mathcal{M}$ when $\mathbf{g}$ is the vector of all ones, i.e., $\mathbf{g} = \mathbf{1} = (1, 1, \ldots, 1) \in \{ 1 \}^N$ or $\mathrm{diag}(\mathbf{g})$ is the $N \times N$ identity matrix. Moreover, the weight of any row of $\mathcal{B}$ is at most $N/2$.

Given $\mathbf{x} = (x_1, \ldots, x_N)^T \in \{ 0, 1\}^N$, we have:
\begin{eqnarray}
\mathbf{y}^\prime &=& \mathcal{B} \otimes \mathbf{x} = \begin{bmatrix}
\mathcal{B}_{1, *} \otimes \mathbf{x} \\
\vdots \\
\mathcal{B}_{k, *} \otimes \mathbf{x}
\end{bmatrix} = 
\begin{bmatrix}
\left( \mathcal{M}_{1, *} \wedge \mathbf{g} \right) \otimes \mathbf{x} \\
\vdots \\
\left( \mathcal{M}_{k, *} \wedge \mathbf{g} \right) \otimes \mathbf{x}
\end{bmatrix} \\
&=& \bigvee_{j = 1}^N x_j g_{j} \mathcal{M}_j = \bigvee_{\substack{j = 1 \\ x_j = 1}}^N g_{j} \mathcal{M}_j = \bigvee_{\substack{j = 1 \\ x_j g_{j} = 1}}^N \mathcal{M}_j.
\label{extSaffron}
\end{eqnarray}
Therefore, the outcome vector $\mathbf{y}^\prime$ is the union of at most $|\mathbf{x}|$ columns of $\mathcal{M}$. Because of the decoding of the SAFFRON scheme, given an input outcome vector, we can either identify the only one defective item presented in it or there is no single defective item in the outcome vector. That means conditions (i) and (ii) in section~\ref{subsub:dec1Overview} holds. Condition (iii) also holds because of the construction of $\mathcal{B}$. Note that even when there is only one defective item $j$, i.e, $|\mathbf{x}| = 1$ and $x_j =1$, the defective item cannot be identified if $g_j = 0$.

\subsubsection{Encoding procedure}
\label{subsub:enc1}
Since matrix $\mathcal{B}$ in Eq.~\eqref{matrixB} is insufficient to identify the defective item in the dilution type 2 model, we are going to design a bigger $K \times N$ matrix $\mathcal{A}$ as follows. For any $\delta > 0$, let denote $c = \frac{2\ln(2\log_2{N}/\delta)}{p_0 \lambda^2}$ for some $\lambda > 0$ and $p_0$ is defined in Eq.~\eqref{p0}. Then the measurement matrix $\mathcal{A}$ is designed as follows:
\begin{eqnarray}
\mathcal{A} &:=& \begin{bmatrix}
\mathcal{B}_{1, *} \\
\vdots \\
\mathcal{B}_{1, *} \\
\hline
\vdots \\
\hline
\mathcal{B}_{k, *} \\
\vdots \\
\mathcal{B}_{k, *}
\end{bmatrix} = \begin{bmatrix}
\mathcal{M}_{1, *} \wedge \mathbf{g}\\
\vdots \\
\mathcal{M}_{1, *} \wedge \mathbf{g} \\
\hline
\vdots \\
\hline
\mathcal{M}_{k, *} \wedge \mathbf{g} \\
\vdots \\
\mathcal{M}_{k, *} \wedge \mathbf{g}
\end{bmatrix}
\begin{array}{l}
\left.
\begin{array}{l}
\ \\
\ \\
\ \\
\end{array}
\right\} \ c \mbox{ times} \\
\quad \ \vdots \\
\left.
\begin{array}{l}
\ \\
\ \\
\ \\
\end{array}
\right\} \ c \mbox{ times}
\end{array}; \ 
\mathcal{A}^\star := 
\begin{bmatrix}
\mathcal{M}_{1, *} \\
\vdots \\
\mathcal{M}_{1, *} \\
\hline
\vdots \\
\hline
\mathcal{M}_{k, *} \\
\vdots \\
\mathcal{M}_{k, *}
\end{bmatrix}
\begin{array}{l}
\left.
\begin{array}{l}
\ \\
\ \\
\ \\
\end{array}
\right\} \ c \mbox{ times} \\
\quad \ \vdots \\
\left.
\begin{array}{l}
\ \\
\ \\
\ \\
\end{array}
\right\} \ c \mbox{ times}
\end{array} \label{matrixA} \\ 
\mathbf{y}_* &:=& \mathcal{A} \tilde{\otimes} \mathbf{x} = \begin{bmatrix}
\mathcal{B}_{1, *} \tilde{\otimes} \mathbf{x} \\
\vdots \\
\mathcal{B}_{1, *} \tilde{\otimes} \mathbf{x}\\
\hline
\vdots \\
\hline
\mathcal{B}_{k, *}  \tilde{\otimes} \mathbf{x}\\
\vdots \\
\mathcal{B}_{k, *} \tilde{\otimes} \mathbf{x}
\end{bmatrix} = \begin{bmatrix}
y_{1}^1 \\
\vdots \\
y_{1}^c \\
\hline 
\vdots \\
\hline 
y_{k}^1 \\
\vdots \\
y_{k}^c
\end{bmatrix} =
\begin{bmatrix}
\\
\mathbf{y}_1 \\
\\
\hline 
\vdots \\
\hline 
\\
\mathbf{y}_k \\
\\
\end{bmatrix}; \quad \mathbf{y}_i = \begin{bmatrix}
y_i^1 \\
\vdots \\
y_i^c
\end{bmatrix} = \begin{bmatrix}
\mathcal{B}_{i, *} \tilde{\otimes} \mathbf{x} \\
\vdots \\
\mathcal{B}_{i, *} \tilde{\otimes} \mathbf{x}
\end{bmatrix} \label{eqn:1Defec}
\end{eqnarray}
where every row $\mathcal{B}_{i, *}$ is repeated $c$ times, and $K = c \times k =  O \left( \frac{\log_2{N} \times \ln(\log_2{N}/\delta)}{p_0 \lambda^2} \right)$ for $i = 1, \ldots, k$. When $\mathbf{g} = \mathbf{1}$, $\mathcal{A}$ is denoted as $\mathcal{A}^\star$. Straightforwardly, $\mathcal{A}^\star$ is nonrandomly constructed. If $\mathbf{g}$ is nonrandomly constructed, then $\mathcal{A}$ is also nonrandomly constructed.

Assume that $\mathbf{y}_* \in \{0, 1 \}^K$ is the outcome vector we observe. Then $\mathbf{y}_* = \mathcal{A} \tilde{\otimes} \mathbf{x}$ as Eq.~\eqref{eqn:1Defec}, where $y_i^l = \mathcal{B}_{i, *} \tilde{\otimes} \mathbf{x} \in \{0, 1 \}$, and $\mathbf{y}_i = \begin{bmatrix}
y_i^1 \\
\vdots \\
y_i^c
\end{bmatrix} \in \{0, 1 \}^c$ for $i = 1, \ldots, k$ and $l = 1, \ldots, c$. Remember that each $y_i^l$ has the probabilities of false positive and false negative such that the probability where the outcome is correct is at least $p_0$. In addition, the weight of every row of $\mathcal{A}$ is at most $N/2$.

This procedure is illustrated as follows:
\begin{eqnarray}
\mathcal{B} \xrightarrow{\text{enlarge}} \mathcal{A} \xrightarrow{\tilde{\otimes} \mathbf{x}} \mathbf{y}_* = \mathcal{A} \tilde{\otimes} \mathbf{x} = \begin{bmatrix}
\mathbf{y}_1 \\
\vdots \\
\mathbf{y}_k
\end{bmatrix} \notag
\end{eqnarray}

\subsubsection{Decoding procedure}
\label{subsub:dec1}
The decoding procedure is described as Algorithm~\ref{alg:decoding1Defect} with descriptions, and illustrated as follows:
\begin{eqnarray}
\mathbf{y}_* = \begin{bmatrix}
\mathbf{y}_1 \\
\vdots \\
\mathbf{y}_k
\end{bmatrix}
\begin{array}{l}
\xrightarrow[\text{\ref{subsub:sumDilution}}]{\text{Section}} y^\prime_1 = \mathcal{B}_{1, *} \otimes \mathbf{x} \\
\ \quad \vdots \\
\xrightarrow[\text{\ref{subsub:sumDilution}}]{\text{Section}} y^\prime_k = \mathcal{B}_{k, *} \otimes \mathbf{x}
\end{array}
\left.
\begin{array}{l}
\\
\\
\\
\end{array}
\right\}
\rightarrow \mathbf{y}^\prime = \mathcal{B} \otimes \mathbf{x}
\xrightarrow[\text{\ref{subsub:extSAFFRON}}]{\text{Section}} \mathrm{wt}(\mathbf{y}^\prime)
\rightarrow \begin{tabular}{@{}c@{}} (No) \\ defective \\ item \end{tabular} \notag
\end{eqnarray}

The function $\mathrm{base10}(\cdot)$ converts the binary input vector into a decimal number. For example, $\mathrm{base10}(1, 0, 1) = 1 \times 2^0 + 0 \times 2^1 + 1 \times 2^2 = 5$. This procedure is briefly described here: step 1 initializes the value of the defective item. If there is no defective item in $\mathbf{y}_*$, the algorithm returns $-1$. Step 4 scans all outcomes. Steps 4--11 recover $y^\prime_i = \mathcal{B}_{i, *} \otimes \mathbf{x}$ from $\mathbf{y}_i$. As steps 4--11 finished running, one gets $\mathbf{y}^\prime = \mathcal{B} \otimes \mathbf{x}$. Step 12 is to check whether there is only one defective item in $\mathbf{y}^\prime$ as described in section~\ref{subsub:extSAFFRON}. Step 13 is to compute the index of the defective item if it is available.

\begin{algorithm}
\caption{$\mathrm{Dec1Defect}(\mathbf{y}_*, N, c)$}
\label{alg:decoding1Defect}
\textbf{Input:} Outcome vector $\mathbf{y}_*$, number of items $N$, a constant $c$.\\
\textbf{Output:} The defective item  $j_0$ (if available).

\begin{algorithmic}[1]
\STATE $j_0 = -1$; \COMMENT{The initial index for the defective item.}
\STATE $k = 2 \times \log_2{N}$; \COMMENT{The number of rows in $\mathcal{M}$.}
\STATE $\mathbf{y}^\prime = (0, \ldots, 0)^T$; \COMMENT{The initial value for $\mathbf{y}^\prime$.}
\FOR [Scan from $\mathbf{y}_1$ to $\mathbf{y}_k$]{$i=1$ \TO $k$}
	\IF [\#1s in $\mathbf{y}_i$ is dominant] {$\mathrm{wt}(\mathbf{y}_i) > \frac{1}{2} c$}
		\STATE $y_i^\prime = 1$. \COMMENT{The true outcome is 1.}
	\ENDIF
	\IF [\#0s in $\mathbf{y}_i$ is dominant] {$c - \mathrm{wt}(\mathbf{y}_i) > \frac{1}{2} c$}
		\STATE $y_i^\prime = 0$. \COMMENT{The true outcome is 0.}
	\ENDIF
\ENDFOR
\IF [There exists only one defective item.] {$\mathrm{wt}(\mathbf{y}^\prime) = \log_2{N}$}
	\STATE $j_0 = \mathrm{base10}(y_1^\prime, \ldots, y_{k/2}^\prime)$. \COMMENT{Identify the defective item.}
\ENDIF
\STATE Return $j_0$. \COMMENT{Return the defective item.}
\end{algorithmic}
\end{algorithm}

\subsection{Decoding $d \geq 2$ defective items}
\label{sub:decodingD}
We present the decoding algorithm for the case of at most $d$ defective items. The basic ideas are: a measurement matrix is created when the number of items in each test does not exceed $N/2$; and there are at least $d$ tests where each contains only one defective item and all defective items belong to them. Each specific test is associated with the matrix $\mathcal{A}$ in Eq. \eqref{matrixA} for efficient identification.

\subsubsection{Encoding procedure}
\label{subsub:encD}

There are two steps in creating a $t \times N$ measurement matrix $\mathcal{T}$. The first step is to create a $h \times N$ matrix $\mathcal{G}$ having the property as follows:
\begin{enumerate}{\roman{enumii}}
\item The row weight is constant and not exceeded $N/2$.
\item For every $d$ columns of $\mathcal{G}$ (representing for $d$ potential defective items), denoted $j_1, \ldots, j_d$, with one designated column, e.g., $j_1$, there exists a row, e.g., $i_1$, such that $g_{i_1 j_1} = 1$ and $g_{i_1 j_a} = 0$, where $a \in \{2, 3, \ldots, d \}$. In short, for every $d$ columns of $\mathcal{G}$, there exists an $d \times d$ identity matrix constructed by those $d$ columns.
\end{enumerate}
The first condition is to make the dilution type 2 model in section~\ref{subsub:sumDilution} hold. The second condition is to guarantee that at most $d$ defective items can be identified. To fulfill these conditions, $\mathcal{G}$ is chosen as a $(d-1)$-disjunct matrix generated in Lemma~\ref{lem:ConstDisjunct}. We note that $h < d^2 \log_2^2{N}$ because of the construction of $\mathcal{G}$ in the Appendix~\ref{app:Thr1}.

Then, the second step is to create a $K \times N$ signature matrix $\mathcal{A}$ such that given an arbitrary column of $\mathcal{A}$, its column index in $\mathcal{A}$ can be efficiently identified. $\mathcal{A}$ is chosen as $\mathcal{A}^\star$ in Eq. \eqref{matrixA} for $c = \frac{2 \ln(2d^2 \log_2^3{N}/\delta)}{p_0 \lambda^2}$, and $K = ck$.

Finally, the measurement matrix $\mathcal{T}$ is created as follows:
\begin{eqnarray}
\label{matrixT}
\mathcal{T} &:=& \begin{bmatrix}
\mathcal{A}^\star \times \mathrm{diag}(\mathcal{G}_{1, *}) \\
\vdots \\
\mathcal{A}^\star \times \mathrm{diag}(\mathcal{G}_{h, *})
\end{bmatrix} = \begin{bmatrix}
\mathcal{A}^1 \\
\vdots \\
\mathcal{A}^h
\end{bmatrix};
\end{eqnarray}
\noindent
where $t = h \times K = h \times c \times k = O \left( \frac{d^2\log_2^3{N} \times \ln(d^2\log_2^3{N}/\delta)}{p_0 \lambda^2} \right)$ and $\mathcal{A}^w := \mathcal{A}^\star \times \mathrm{diag}(\mathcal{G}_{w, *}) \in \{0, 1 \}^{K \times N}$for $w = 1, \ldots, h$. We note that each $\mathcal{A}^w$ is created by setting $\mathbf{g} = \mathcal{G}_{w, *}$ in Eq.~\eqref{matrixA}. It is equivalent that $\mathcal{A}^w$ is created from $\mathcal{B}^w = \mathcal{M} \times \mathrm{diag}(\mathcal{G}_{w, *})$, which is achieved by setting $\mathbf{g} = \mathcal{G}_{w, *}$ in Eq.~\eqref{matrixB}. The precise formulas of $\mathcal{B}^w$ and $\mathcal{A}^w$ are defined in Eq.~\eqref{Bi} and Eq.~\eqref{Ai}.
\begin{eqnarray}
\mathcal{B}^w &:=& \mathcal{M} \times \mathrm{diag}(\mathcal{G}_{w, *}) = \begin{bmatrix}
\mathcal{M}_{1, *} \wedge \mathcal{G}_{w, *} \\
\vdots \\
\mathcal{M}_{k, *} \wedge \mathcal{G}_{w, *}
\end{bmatrix} = \begin{bmatrix}
\mathcal{B}^w_{1, *} \\
\vdots \\
\mathcal{B}^w_{k, *}
\end{bmatrix} \label{Bi} \\
\mathcal{A}^w &:=& \mathcal{A}^\star \times \mathrm{diag}(\mathcal{G}_{w, *}) = \begin{bmatrix}
\mathcal{B}^w_{1, *} \\
\vdots \\
\mathcal{B}^w_{1, *} \\
\hline
\vdots \\
\hline
\mathcal{B}^w_{k, *} \\
\vdots \\
\mathcal{B}^w_{k, *}
\end{bmatrix}
\begin{array}{l}
\left.
\begin{array}{l}
\ \\
\ \\
\ \\
\end{array}
\right\} \ c \mbox{ times} \\
\quad \ \vdots \\
\left.
\begin{array}{l}
\ \\
\ \\
\ \\
\end{array}
\right\} \ c \mbox{ times}
\end{array} \label{Ai}
\end{eqnarray}

Since the row weight of $\mathcal{G}_{w, *}$ does not exceed $N/2$, the weight of every row in $\mathcal{A}^w$ does not exceed $N/2$. Therefore, the weight of every row in $\mathcal{T}$ does not exceed $N/2$. In addition, $\mathcal{T}$ is nonrandomly constructed because $\mathcal{A}^\star$ and $\mathcal{G}$ are nonrandomly constructed.

The outcome of tests using the measurement matrix $\mathcal{T}$ is:
\begin{equation}
\label{OutcomeMatrixT}
\mathbf{y} = \mathcal{T} \tilde{\otimes} \mathbf{x} = \begin{bmatrix}
\mathcal{A}^1 \tilde{\otimes} \mathbf{x} \\
\vdots \\
\mathcal{A}^h \tilde{\otimes} \mathbf{x}
\end{bmatrix} = \begin{bmatrix}
\mathbf{y}_*^1 \\
\vdots \\
\mathbf{y}_*^h
\end{bmatrix} \in \{0, 1 \}^t
\end{equation}
where
\begin{equation}
\mathbf{y}_*^w := \mathcal{A}^w \tilde{\otimes} \mathbf{x} = 
\begin{bmatrix}
\mathcal{B}^w_{1, *} \tilde{\otimes} \mathbf{x} \\
\vdots \\
\mathcal{B}^w_{1, *} \tilde{\otimes} \mathbf{x} \\
\hline
\vdots \\
\hline
\mathcal{B}^w_{k, *} \tilde{\otimes} \mathbf{x} \\
\vdots \\
\mathcal{B}^w_{k, *} \tilde{\otimes} \mathbf{x}
\end{bmatrix} 
\begin{array}{l}
\left.
\begin{array}{l}
\ \\
\ \\
\ \\
\end{array}
\right\} \ c \mbox{ times} \\
\quad \ \vdots \\
\left.
\begin{array}{l}
\ \\
\ \\
\ \\
\end{array}
\right\} \ c \mbox{ times}
\end{array} = \begin{bmatrix}
\\
\mathbf{y}_1^w \\
\\
\hline 
\vdots \\
\hline 
\\
\mathbf{y}_k^w \\
\\
\end{bmatrix}
\in \{0, 1 \}^K
\label{yi}
\end{equation}
and
\begin{equation}
\mathbf{y}_i^w = \begin{bmatrix}
\mathcal{B}^w_{i, *} \tilde{\otimes} \mathbf{x} \\
\vdots \\
\mathcal{B}^w_{i, *} \tilde{\otimes} \mathbf{x}
\end{bmatrix}
\end{equation}
for $i = 1, \ldots, k$ and $w = 1, \ldots, h$.

This procedure is illustrated as follows:
\begin{eqnarray}
\mathcal{A}^\star, \mathcal{G} \rightarrow \mathcal{T} = \begin{bmatrix}
\mathcal{A}^\star \times \mathrm{diag}(\mathcal{G}_{1, *}) \\
\vdots \\
\mathcal{A}^\star \times \mathrm{diag}(\mathcal{G}_{h, *})
\end{bmatrix}
\xrightarrow{\tilde{\otimes} \mathbf{x}} \mathbf{y} = \mathcal{T} \tilde{\otimes} \mathbf{x} = \begin{bmatrix}
\mathbf{y}_*^1 \\
\vdots \\
\mathbf{y}_*^h
\end{bmatrix} \notag
\end{eqnarray}

\subsubsection{Decoding procedure}
\label{subsub:decD}
To identify the defective items, one simply decodes each $\mathbf{y}_*^w$ by using Algorithm~\ref{alg:decoding1Defect}, for $w = 1, \ldots, h$. The decoding procedure is summarized in Algorithm~\ref{alg:decodingDDefect} and illustrated here:
\begin{eqnarray}
\mathbf{y} = \begin{bmatrix}
\mathbf{y}_*^1 \\
\vdots \\
\mathbf{y}_*^h
\end{bmatrix}
\begin{array}{l}
\xrightarrow{\mathrm{Dec1Defect}(\mathbf{y}_*^1, N, c)} \text{(No) defective item} \\
\qquad \qquad \vdots \\
\xrightarrow{\mathrm{Dec1Defect}(\mathbf{y}_*^h, N, c)} \text{(No) defective item}
\end{array} \notag
\end{eqnarray}

\begin{algorithm}
\caption{Decoding at most $d$ defective items}
\label{alg:decodingDDefect}
\textbf{Input:} Outcome vector $\mathbf{y}$, number of items $N$, a constant $c$, $h$.\\
\textbf{Output:} Set of defective items.

\begin{algorithmic}[1]
\STATE $S = \emptyset$. \COMMENT{Create an empty set of defective items.}
\FOR[Scan all outcome vectors]{$w = 1$ \TO $h$} 
	\STATE $index = \mathrm{Dec1Defect}(\mathbf{y}_*^w, N, c)$. \COMMENT{Find the defective item in $\mathbf{y}_*^w$.}
	\IF[If there is a defective item in $\mathbf{y}_*^w$.]{$index \neq -1$} 
		\STATE $S = S + \{ index \}$.  \COMMENT{Add the defective item into the defective set.}
	\ENDIF
\ENDFOR

\STATE Return $S$. \COMMENT{Return the defective set.}
\end{algorithmic}
\end{algorithm}

\section{Main results}
\label{sec:main}

\subsection{Decoding a defective item}
\label{sub:main1}

Assume that there is only one defective item in $N$ items. Since the model in the section~\ref{subsub:Dilu2Equal} is an instance of the model in the section~\ref{subsub:sumDilution}, the following theorem is to summarize the encoding and decoding procedures in section~\ref{sub:decoding1}:
\begin{theorem}
\label{thr:1Defect}
Let $1 < N$ be an integer, $0 \leq \theta_{0}, \theta_{1} \leq \frac{1}{2}$ be real scalars, and $0 < \lambda, \xi$. Suppose that $\theta_0 + \theta_1 \leq 2 \left( 1 - \frac{1/2 + \xi}{1 - \lambda} \right)$. Assume that there is at most one defective item in $N$ items. And noise in the test outcome depends on the number of items in a test under dilution type 2 model in the section~\ref{subsub:Dilu2Equal}, in which the probabilities of a false positive and of a false negative are $\theta_0$ and $\theta_1$, and the number of items in every test is $N/2$. For any $\delta > 0$, there exists a nonrandomly constructed $K \times N$ matrix, where $K = O \left( \frac{\log_2{N} \times \ln(\log_2{N}/\delta)}{\left( 1 - \frac{1}{2}(\theta_0 + \theta_1) \right) \lambda^2} \right)$, such that a defective item can be identified in time $O(K)$ with probability at least $1 - \delta$.
\end{theorem}

\begin{proof}
Let $\mathcal{A}^\star$ in Eq. \eqref{matrixA} be the measurement matrix. Then $\mathcal{B} = \mathcal{M}$ and $\mathbf{g} = \mathbf{1}$ in Eq.~\eqref{Bi}. Then $\mathcal{A}^\star$ can be nonrandomly constructed because of its construction. Since the Hamming weight of its row is $N/2$, the number of items in each test using $\mathcal{A}^\star$ is equal to $N/2$. On the other hand, $\theta_0 + \theta_1 \leq 2 \left( 1 - \frac{1/2 + \xi}{1 - \lambda} \right)$. Therefore, the dilution type 2 model in the section~\ref{subsub:Dilu2Equal} can be applied here.

From the construction of $\mathcal{A}^\star$, we have
\begin{equation}
K = O \left( \frac{\log_2{N} \times \ln(\log_2{N}/\delta)}{p_0 \lambda^2} \right) =
O \left( \frac{\log_2{N} \times \ln(\log_2{N}/\delta)}{\left( 1 - \frac{1}{2}(\theta_0 + \theta_1) \right) \lambda^2} \right).
\end{equation}

We start analyzing Algorithm~\ref{alg:decoding1Defect} now. Steps 5--10 are to recover the correct outcome. Using Eq. \eqref{HerculeanChernoff}, the probability that $y_i^\prime = \mathcal{B}_{i, *} \otimes \mathbf{x}$ is wrongly derived from $\mathbf{y}_i = \begin{bmatrix}
\mathcal{B}_{i, *} \tilde{\otimes} \mathbf{x} \\
\vdots \\
\mathcal{B}_{i, *} \tilde{\otimes} \mathbf{x}
\end{bmatrix}$ is at most:
\begin{eqnarray}
\mathrm{exp} \left(- \frac{\lambda^2 \mu_0}{2} \right) \leq \mathrm{exp} \left( - \frac{\lambda^2}{2} \cdot \frac{2}{\lambda^2} \cdot \ln \frac{2\log_2{N}}{\delta} \right) = \frac{\delta}{2\log_2{N}}. \label{cmu0}
\end{eqnarray}

Since $i = 1,\ldots, k$, using union bound, the total probability that Algorithm~\ref{alg:decoding1Defect} fails to recover \textbf{any} $y_i^\prime = \mathcal{B}_{i, *} \otimes \mathbf{x}$, i.e., recover $\mathbf{y}^\prime = \mathcal{B} \otimes \mathbf{x}$, is at most:
\begin{equation}
k \times \frac{\delta}{2\log_2{N}} = 2\log_2{N} \times \frac{\delta}{2\log_2{N}} = \delta.
\label{failDec1}
\end{equation}
Thus, Algorithm~\ref{alg:decoding1Defect} can recover the defective item with probability at least $1 - \delta$. Since $\mathbf{y}^\prime = \mathcal{B} \otimes \mathbf{x} = \mathcal{M} \otimes \mathbf{x}$, the defective item can be identified by using the decoding procedure in the section~\ref{subsub:SAFFRON}. Since we scan the test outcomes once, the decoding complexity of the algorithm is $K = O(K)$. $\blacksquare$
\end{proof}

\subsection{Decoding $d \geq 2$ defective items}
\label{sub:mainD}

From the encoding and decoding procedures in the section~\ref{sub:decodingD}, the probability of successfully identifying at most $d$ defective items is stated as follows:

\begin{theorem}
\label{thr:DDefect}
Let $1 < d < N$ be integers, $0 \leq \theta_{0}, \theta_{1} \leq \frac{1}{2}$ be real scalars, and $0 < \lambda, \xi$. Assume that $\theta_0 + \theta_1 \leq 2 \left( 1 - \frac{1/2 + \xi}{1 - \lambda} \right)$. Noise in the test outcome depends on the number of items in that test under dilution type 2 model in the section~\ref{subsub:sumDilution} in which the number of items in a test is up to $N/2$, and the probability where an outcome is correct is at least $p_0 = 1 - \frac{1}{2}(\theta_0 + \theta_1)$. For any $\delta > 0$, there exists a nonrandomly constructed $t \times N$ matrix $\mathcal{T}$, where $t = O \left( \frac{d^2\log_2^3{N} \times \ln(d^2\log_2^3{N}/\delta)}{\left( 1 - \frac{1}{2}(\theta_0 + \theta_1) \right) \lambda^2} \right)$, such that at most $d$ defective item can be identified in time $O(t)$ with probability at least $1 - \delta$.
\end{theorem}

\begin{proof}
From the construction of matrix $\mathcal{T}$ in the section~\ref{subsub:encD}, it is nonrandomly constructed. The number of tests is $t = O \left( \frac{d^2\log_2^3{N} \times \ln(d^2\log_2^3{N}/\delta)}{\left( 1 - \frac{1}{2}(\theta_0 + \theta_1) \right) \lambda^2} \right)$ and the weight of each row in $\mathcal{T}$ is up to $N/2$. That means the number of items in each test in $\mathcal{T}$ does not exceed $N/2$. In addition, $\theta_0 + \theta_1 \leq 2 \left( 1 - \frac{1/2 + \xi}{1 - \lambda} \right)$. Then, noise in the test outcome depends on the number of items in that test under dilution type 2 model in the section~\ref{subsub:sumDilution}.

Step 3 in Algorithm~\ref{alg:decodingDDefect} is interpreted here. Because $\mathbf{y}_*^w$ is defined as Eq.~\eqref{yi}, the function $\mathrm{Dec1Defect}(\mathbf{y}_*^w, N, c)$ first recovers $\mathbf{y}^\prime_w$ from $\mathbf{y}_*^w$ (Steps 4--11 in Algorithm~\ref{alg:decoding1Defect}), where
\begin{equation}
\mathbf{y}^\prime_w = \mathcal{B}^w \otimes \mathbf{x} = \begin{bmatrix}
\mathcal{B}^w_{1, *} \otimes \mathbf{x} \\
\vdots \\
\mathcal{B}^w_{k, *} \otimes \mathbf{x}
\end{bmatrix} = \bigvee_{\substack{j = 1 \\ x_j g_{wj} = 1}}^N \mathcal{M}_j.
\label{ExtremelyComplicated}
\end{equation}
We note that $\mathcal{B}^w$ is the matrix $\mathcal{B}$ in Eq.~\eqref{matrixB} by setting $\mathbf{g} = \mathcal{G}_{w, *}$. Matrix $\mathcal{B}^w$ is used to generate the matrix $\mathcal{A}^w$ in Eq.~\eqref{Ai}. Then $\mathbf{y}^\prime_w$ is decoded as described in the section~\ref{subsub:extSAFFRON}. The decoding complexity of step 3 in Algorithm~\ref{alg:decodingDDefect} is $O(ck)$.

We are going to estimate the probability that Algorithm~\ref{alg:decoding1Defect} fails to recover $\mathbf{y}^\prime_w$ from $\mathbf{y}_*^w$. Since $c = \frac{2 \ln(2d^2 \log_2^3{N}/\delta)}{p_0 \lambda^2}$, we have $\mu_0 = p_0 c = \frac{2 \ln(2d^2 \log_2^3{N}/\delta)}{\lambda^2} $. Then, from Eq.~\eqref{HerculeanChernoff}, the probability that $\mathcal{B}^w_{i, *} \otimes \mathbf{x}$ is not recovered from $c$ trials of test $\mathcal{B}^w_{i, *} \tilde{\otimes} \mathbf{x}$ in $\mathbf{y}_*^w$ is at most:
\begin{eqnarray}
\mathrm{exp} \left(- \frac{\lambda^2 \mu_0}{2} \right) \leq \mathrm{exp} \left( - \frac{\lambda^2}{2} \cdot \frac{2 \ln(2d^2 \log_2^3{N}/\delta)}{\lambda^2} \right) = \frac{\delta}{2d^2 \log_2^3{N}} \notag
\end{eqnarray}
\noindent
where $i = 1, \ldots, k$,

Since there is $k = 2\log_2{N}$, the probability that Algorithm~\ref{alg:decoding1Defect} fails to recover $\mathbf{y}^\prime_w$ from $\mathbf{y}_*^w$ is at most:
\begin{eqnarray}
k \times \frac{\delta}{2d^2 \log_2^3{N}} = 2\log_2{N} \times \frac{\delta}{2d^2 \log_2^3{N}} = \frac{\delta}{d^2 \log_2^2{N}}
\end{eqnarray}

From Algorithm~\ref{alg:decodingDDefect}, since $w = 1, \ldots, h$ and $h < d^2 \log_2^2{N}$, using union bound, the total probability that the algorithm fails to decode any $\mathbf{y}_*^1, \ldots, \mathbf{y}_*^k$ is at most:
\begin{equation}
h \times \frac{\delta}{d^2 \log_2^2{N}} < d^2 \log_2^2{N} \times \frac{\delta}{d^2 \log_2^2{N}} = \delta
\end{equation}

Let $\mathbb{G} = \{j_1, \ldots, j_{|\mathbb{G}|} \}$ be the defective set, where $|\mathbb{G}| \leq d$. Our task now is to prove that there exists $\mathbf{y}^\prime_{i_1}, \ldots, \mathbf{y}^\prime_{i_{|\mathbb{G}|}}$ such that $j_1 = \mathrm{Dec1Defect}(\mathbf{y}^\prime_{i_1}, N, c), \ldots$, $j_1 = \mathrm{Dec1Defect}(\mathbf{y}^\prime_{i_{|\mathbb{G}|}}, N, c)$. Since the role of each element of $\mathbb{G}$ is equivalent, we only need proving that there exists row $i_1$ of $\mathcal{G}$ such that $j_1 = \mathrm{Dec1Defect}(\mathbf{y}^\prime_{i_1}, N, c)$. The rest of the element in $\mathbb{G}$ can be identified in the same way.

Indeed, because of the property 2 of $\mathcal{G}$ in the section~\ref{subsub:encD}, for column $j_1$, there exists row $i_1$ such that $g_{i_1j_1} = 1$ and $g_{i_1 j_a} = 0$, where $a \in \{2, 3, \ldots, |\mathbb{G}| \}$. Then Eq.~\eqref{ExtremelyComplicated} becomes:
\begin{eqnarray}
\mathbf{y}_{i_1}^\prime = \bigvee_{\substack{j = 1 \\ x_j g_{i_1 j} = 1}}^N \mathcal{M}_j = \bigvee_{\substack{j \in \mathbb{G} \\ x_j g_{i_1 j} = 1}}^N \mathcal{M}_j = \mathcal{M}_{j_1}.
\end{eqnarray}
Therefore, $j_1 = \mathrm{Dec1Defect}(\mathbf{y}^\prime_{i_1}, N, c)$. Thus, there exists $\mathbf{y}^\prime_{i_1}, \ldots, \mathbf{y}^\prime_{i_{|\mathbb{G}|}}$ such that $j_1 = \mathrm{Dec1Defect}(\mathbf{y}^\prime_{i_1}, N, c), \ldots, j_{|\mathbb{G}|} = \mathrm{Dec1Defect}(\mathbf{y}^\prime_{i_{|\mathbb{G}|}}, N, c)$. Consequently, Algorithm~\ref{alg:decodingDDefect} can identify the defective items with probability at least $1 - \delta$.

For decoding complexity, because Algorithm~\ref{alg:decodingDDefect} scans the test outcome $\mathbf{y}$ once and the decoding complexity of the step 3 is $O(ck)$, the decoding complexity is $h \times O(ck) = O(t)$. $\blacksquare$
\end{proof}

\section{Evaluations}
\label{sec:eval}

\subsection{Overview}
\label{sub:overEval}
We summarized our scheme compared with those in~\cite{atia2012boolean,lee2016saffron}. The results are shown in Table~\ref{tbl:comparison}, where $\times$ means that it is not considered, ``Indept.'' and ``Dept.'' stand for ``Independent'' and ``Dependent''. The notations in Table~\ref{tbl:comparison} are described in sections~\ref{sub:model} and~\ref{sub:new}.

\begin{center}
\begin{table*}[ht]
\centering
\caption{Comparison with existing work.}
\scalebox{0.75}{
\begin{tabular}{|c|c|c|c|c|c|c|}
\hline
Model & \begin{tabular}{@{}c@{}}  Noisy \\ type  \end{tabular} & \begin{tabular}{@{}c@{}} Defective \\ items $(d)$ \end{tabular} & \begin{tabular}{@{}c@{}} Tests \\ $t$ \end{tabular} & \begin{tabular}{@{}c@{}} Decoding \\ complexity \end{tabular} & \begin{tabular}{@{}c@{}} Decoding \\ type \end{tabular} & \begin{tabular}{@{}c@{}} Construction \\ type \end{tabular} \\
\hline
\begin{tabular}{@{}c@{}} Additive noise \\ type 1~\cite{atia2012boolean} \end{tabular} & Indept. & $\geq 1$ & $O\left( \frac{d^2 \log{N}}{1 - \epsilon_1} \right)$ & $\times$ & $\times$ & $\times$ \\
\hline
\begin{tabular}{@{}c@{}} Additive noise \\ type 2~\cite{lee2016saffron} \end{tabular} & Indept. & $\geq 1$ & $O \left( \frac{6 c(\epsilon_2) d \log_2{N}}{1 - H(\epsilon_2) - \delta} \right)$ & $O(t)$ & Random & Random\\
\hline
\begin{tabular}{@{}c@{}} Dilution \\ type 1~\cite{atia2012boolean} \end{tabular} & Indept. & $\geq 1$ & $O \left( \frac{d^2 \log_2{N}}{(1-\epsilon_3)^2} \right)$ & $\times$ & $\times$ & $\times$\\
\hline
\begin{tabular}{@{}c@{}} \textbf{Dilution type 2} \\ \textbf{(Proposed)} \end{tabular} & \textbf{Dept.} & \begin{tabular}{@{}c@{}} 1 \vspace{2mm} \\ $\geq 2$ \end{tabular} & \begin{tabular}{@{}c@{}} $O \left( \frac{\log_2{N} \times \log_2(\log_2{N}/\delta)}{\left( 1 - \frac{1}{2}(\theta_0 + \theta_1) \right) \lambda^2} \right)$ \\ $O \left( \frac{d^2\log_2^3{N} \times \log_2(d^2\log_2^3{N}/\delta)}{\left( 1 - \frac{1}{2}(\theta_0 + \theta_1) \right) \lambda^2} \right)$ \end{tabular}  & $O(t)$ & \textbf{Random} & \textbf{Nonrandom} \\
\hline
\end{tabular}}

\label{tbl:comparison}
\end{table*}

\end{center}

Since the noise type in each model is different in the sections~\ref{sub:model} and~\ref{sub:new}, it is \textbf{incompatible} to evaluate which model is better. Therefore, we demonstrate our simulations from the view point of the number of tests needed, accuracy, and decoding time. All the decoding times are the averages of a hundred runs.

Our proposed schemes are outperformed to the existing schemes in term of construction of measurement matrices. Because our proposed construction is nonrandom, there is \textit{no need in storing} the measurement matrices. On the contrary, the rest of the existing schemes is random, in which each entry is generated randomly and measurement matrices needed to be stored before implementing tests.

\subsection{Parameter settings}
\label{sub:paraSetup}

Let $\lambda = \frac{1}{3}$ and $\xi = 0.001$. $\theta_0$ and $\theta_1$ are chosen such that Eq.~\eqref{conditionTheta} holds, i.e., $\theta_0 + \theta_1 \leq 0.4970$. In particular, $\theta_0$ is up to 0.2 and $\theta_1$ is up to 0.1. Therefore, there are five parameters left to be considered here. These are the number of items $N$, the maximum number of defective items $d$, the precision of decoding algorithms $\delta$, the probability of false positive $\theta_0$, and the probability of false negative $\theta_1$. We fix the precision $\delta = 0.001$ to reduce the number of experiments.

First, $N = \{2^{21}, 2^{28}, 2^{30}, 2^{32}, 2^{33} \}$ is the collection of the number of items used here. When $d \geq 2$, the number of defective items used are $3, 6,$ and $16$. $\theta_0$ and $\theta_1$ are set up in each simulation.

The $h \times N$ $(d-1)$-disjunct matrix $\mathcal{G}$ is generated in Lemma~\ref{lem:ConstDisjunct}. As we mentioned in the proof in Appendix~\ref{app:Thr1}, $\mathcal{G}$ is generated from a $[n, r, n - r + 1]_q$ MDS code $C$ as a Reed-Solomon code~\cite{reed1960polynomial} and an $q \times q$ identity matrix. Note that $N = q^r$, $h = qn$, and $d - 1 = \lfloor \frac{n-1}{r-1} \rfloor$. Table~\ref{tbl:setting} shows 5 instances for Reed-Solomon code, which are used in our simulations. For each $n$, there are a corresponding $d-1$ and $h$. For example, in case $\langle 1 \rangle$, if $n = 5$ then $d-1 = 2$ and $h = qn = 640$, $n = 11$ then $d-1 = 5$ and $h = qn = 1,408$. We note that for a fix value $d-1$, there are many classes of Reed-Solomon code satisfying condition $d -1 = \lfloor \frac{n-1}{r-1} \rfloor$ as in Table~\ref{tbl:setting}.
\begin{table}[ht]
\caption{Settings for $h \times N$ $(d-1)$-disjunct matrix $\mathcal{G}$.}
\begin{center}
\scalebox{0.95}{
\begin{tabular}{|c|c|c|c|c|c|c|}
\hline
Case & $q$ & $n$ & $r$ & $d-1$ & $h = qn$ & $N = q^r$ \\
\hline
$\langle 1 \rangle$ & 128 & $\{5; 11; 31 \}$ & 3 & $\{2; 5; 15 \}$ & $\{640; 1,408; 3,968\}$ & \begin{tabular}{@{}c@{}} $128^3 = 2^{21}$ \\ $= 2,097,152$ \end{tabular} \\
\hline
$\langle 2 \rangle$ & 128 & $\{7; 16; 46 \}$ & 4 & $\{2; 5; 15 \}$ & $\{896; 2,048; 5,888 \}$ & \begin{tabular}{@{}c@{}} $128^4 = 2^{28}$ \\ $= 268,435,456$ \end{tabular} \\
\hline
$\langle 3 \rangle$ & 64 & $\{11; 21; 61 \}$ & 5 & $\{2; 5; 15 \}$ & $\{704; 1,344; 3,904 \}$ & \begin{tabular}{@{}c@{}} $64^5 = 2^{30}$ \\ $= 1,073,741,824$ \end{tabular} \\
\hline
$\langle 4 \rangle$ & 256 & $\{7; 16; 46 \}$ & 4 & $\{2; 5; 15 \}$ & $\{1,792; 4,096; 11,776 \}$ & \begin{tabular}{@{}c@{}} $256^4 = 2^{32}$ \\ $= 4,294,967,296$ \end{tabular} \\
\hline
$\langle 5 \rangle$ & 2048 & $\{5; 11; 31 \}$ & 3 & $\{2; 5; 15 \}$ & $\{10,240; 22,528; 63,488 \}$ & \begin{tabular}{@{}c@{}} $2048^3 = 2^{33}$ \\ $= 8,589,934,592$ \end{tabular} \\
\hline
\end{tabular}}

\label{tbl:setting}
\end{center}
\end{table}

We run simulators for our proposed schemes in Matlab R2015a and test them on a HP Compaq Pro 8300SF with 3.4 GHz Intel Core i7-3770 and 16 GB memory.

\subsection{Number of tests and accuracy}
\label{sub:tests}

When $d=1$, the number of tests in our simulations is $t = 2\log_2{N} \times \lceil \frac{18 \ln(2 \log_2{N}/\delta)}{p_0} \rceil$. $(\theta_0, \theta_1)$ is set to be $(0.002, 0.001), (0.02, 0.01), (0.05, 0.02)$, and $(0.2, 0.1)$. Fig.~\ref{fig:test1} shows that the numbers of tests in these cases do not exceed $16,000$, even when $N = 2^{33} \approx 9$ billion. The graph has the shape of the \textit{logarithmic} function, which is likely linear to $\log_2{N}$. It is also shown that the larger the sum of $\theta_0$ and $\theta_1$, the larger the number of tests.

\begin{center}
\begin{figure}[ht]
\centering
  \includegraphics[scale=0.25]{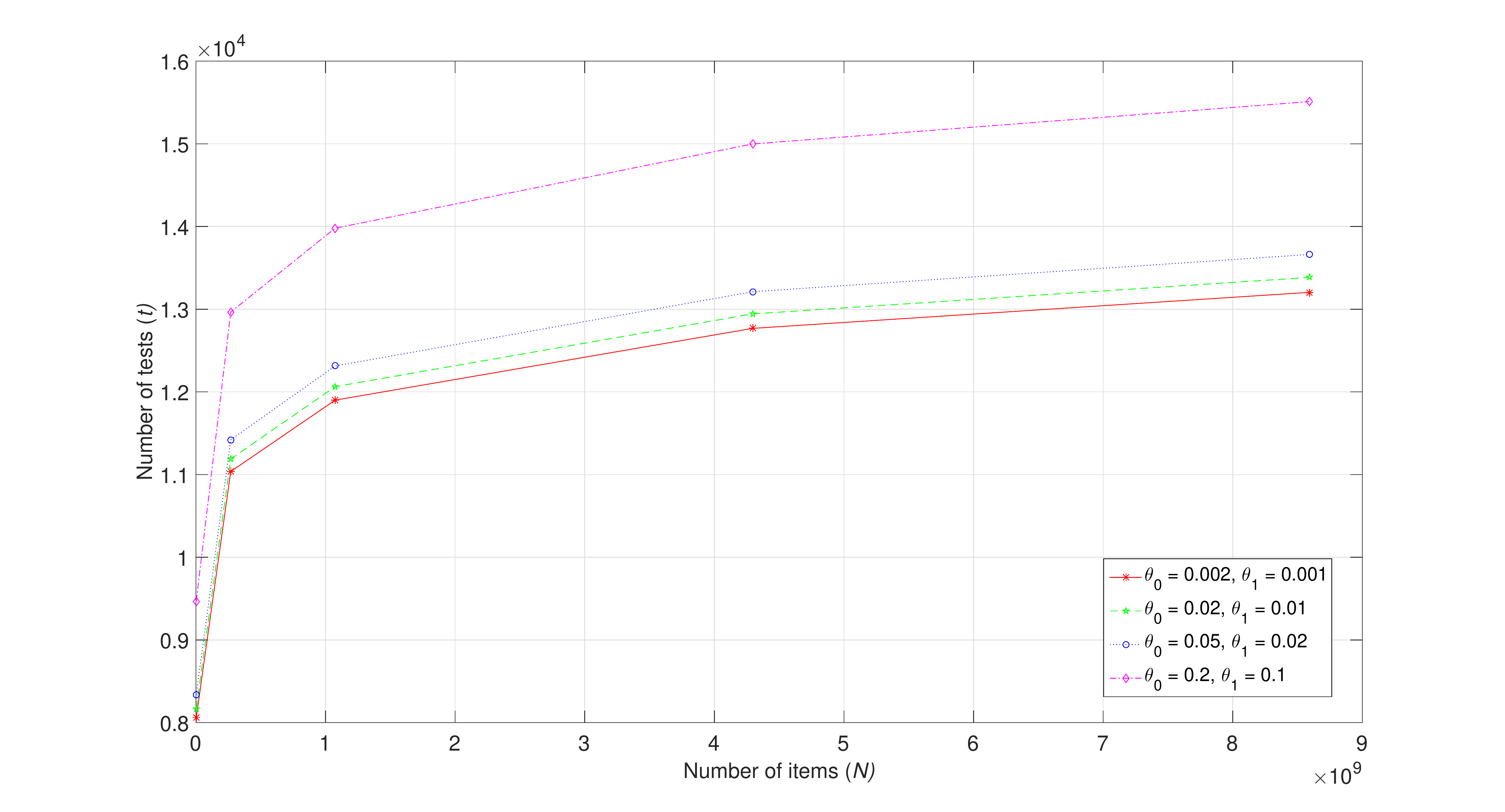}    
  \caption{Number of tests for $d = 1$ when varying $N$, $\theta_0$, and $\theta_1$.}
  \label{fig:test1} 
\end{figure}
\end{center}

When $d \geq 2$, the number of tests in our simulations is $t = 2\log_2{N} \times qn \times \lceil \frac{18 \ln(2d^2 \log_2^3{N}/\delta)}{p_0} \rceil$. We set $(\theta_0, \theta_1) = \{ (0.002, 0.001), (0.2, 0.1)\}$ and $d = \{ 3, 6, 16 \}$. Fig.~\ref{fig:testD} shows that the numbers of tests in these cases are scaled to $N$, $d$, $\theta_0$, and $\theta_1$ in the shape of a \textit{quadratic} function. It is likely to fit to our formula. Moreover, the number of tests is still small (at most $2.5$ billion tests) comparing to the number of items, which is at most $N = 2^{33} \approx 9$ billion.

\begin{figure}[ht]
  \centering
  
  \includegraphics[scale=0.25]{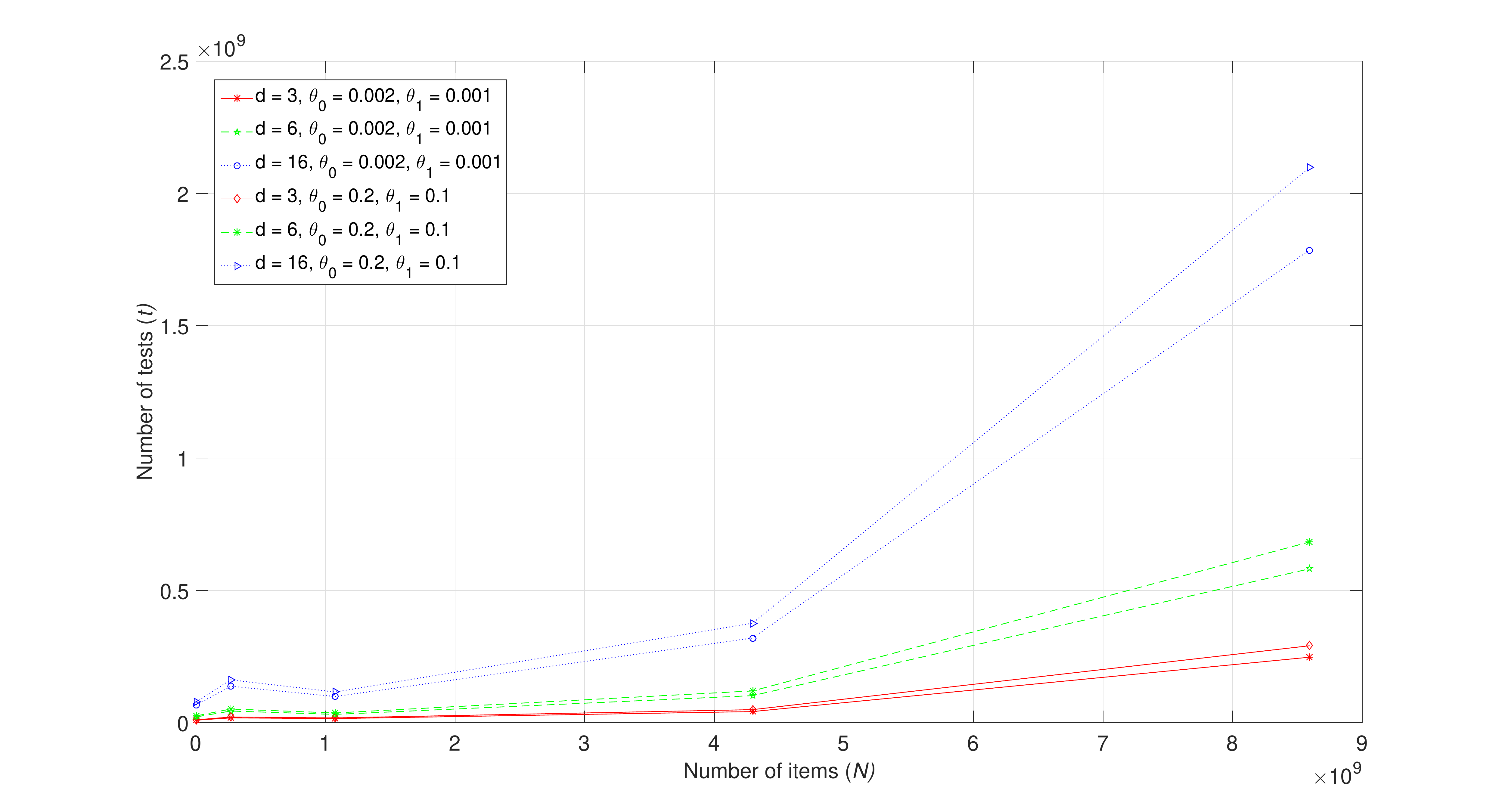}
    
  \caption{Number of tests for $d = 3, 6, 16$ when varying $N$, $d$, $\theta_0$, and $\theta_1$.}
  \label{fig:testD}
\end{figure}

It is noted that all defective items $(d \leq 16)$ are identified with probability 1 for all simulations.

\subsection{Decoding time}
\label{sub:decTime}
We also set $(\theta_0, \theta_1) = \{ (0.002, 0.001), (0.2, 0.1) \}$ in this simulation. When $d = 1$, the running time is always less than 6 microsecond because the number of tests is smaller than 16,000. Therefore, we only show the decoding time when $d \geq 2$, specifically, $d = \{3, 6, 16 \}$. Experimental results in Fig.~\ref{fig:decD} show that the decoding time does not exceed 7 seconds even when $N = 2^{33} \approx 9$ billion. Moreover, the decoding time is linear to the number of tests.

Since the number of tests is proportional to the sum $\theta_0 + \theta_1$, it increases as $\theta_0 + \theta_1$ increases. Then the decoding time should increase as $\theta_0 + \theta_1$ increases because it is linear to the number of tests. However, Fig.~\ref{fig:decD} somehow shows that the decoding time is \textit{increasing} when $\theta_0 + \theta_1$ slightly \textit{decreases}. We can explain this phenomenon by looking at the construction of $\mathcal{T}$. There likely are more operations executed in the step 3 in Algorithm~\ref{alg:decodingDDefect} when $p_0$ decreases. Since the step 3 in Algorithm~\ref{alg:decodingDDefect} calls Algorithm~\ref{alg:decoding1Defect}, we start analyzing Algorithm~\ref{alg:decoding1Defect} here. Because of the construction of $\mathcal{G}$, the number of rows that satisfies the condition 2 in the section~\ref{subsub:encD} likely decreases when $p_0$ increases. Consequently, there are more operations executed at the step 13 in Algorithm~\ref{alg:decoding1Defect}. Therefore, it is not a disadvantage that the number of tests sometimes increases. However, in general, there is a huge difference in decoding time when $p_0$ is close to zero comparing to the case where $p_0$ is close 1.

\begin{figure}[ht] 
  \centering
  
  \includegraphics[scale=0.25]{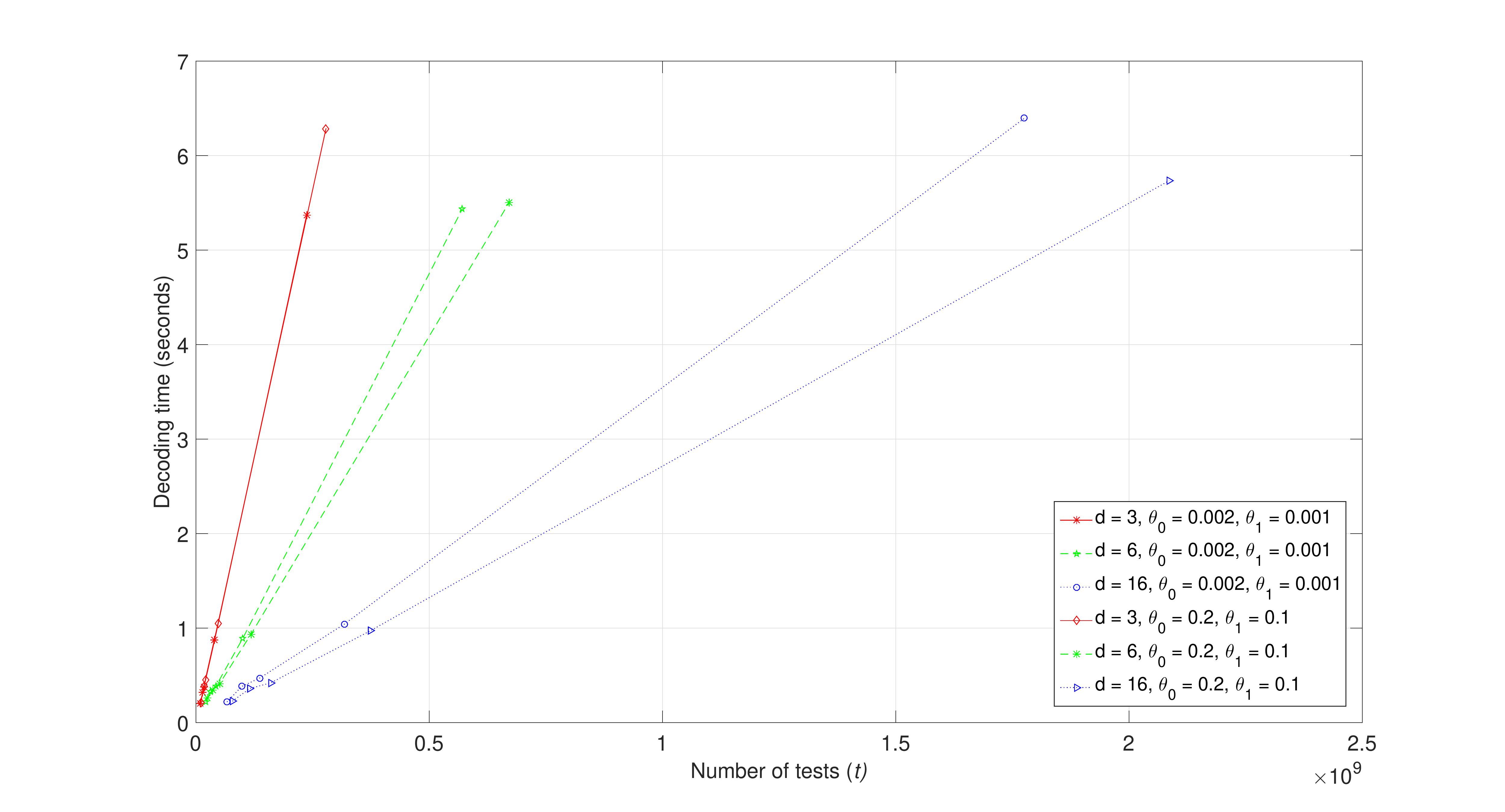}
    
  \caption{Decoding time when varying $N$, $d$, $\theta_0$, and $\theta_1$.}
  \label{fig:decD}
\end{figure}

\section{Conclusion}
\label{sec:cls}
We have proposed efficient schemes to identify defective items in noisy NAGT in which the outcome for a test depends on the number of items it contains. The number of tests in this model is not scaled to $O(d^2 \log{N})$ as usual, but is slightly higher. However, the decoding complexity of our proposed schemes is linear to the number of tests. The open problem is whether exists a scheme that fits to the dilution type 2 model by using $O(d^2 \log{N})$ tests to identify at most $d$ defective items in time $O(d^2 \log{N})$.

\begin{acknowledgements}
The authors thank Teddy Furon for his valuable discussions, Matthias Qusenbauer and Nguyen Son Hoang Quoc for their helpful comments.
\end{acknowledgements}


\bibliographystyle{spmpsci}
\bibliography{bibli}

\appendix
\section{Proof of Lemma \ref{lem:ConstDisjunct}}
\label{app:Thr1}

There are two concepts needed to be introduced before going to prove this lemma. The first one is maximum distance separable code and the second one is concatenation technique. We first brieftly introduce the notion of maximum distance separable (MDS) code. We recommend the reader reading the text book~\cite{guruswamiessential} for further reading. Consider an $[n, r, \Delta]_q$-code $C$ where $n, r$, $\Delta$, and $q$ are block length, dimension, minimum distance, alphabet size, respectively, which satisfy $1 \leq r \leq n < q$. All arithmetic for $C$ is done in the Galois field $\mathbb{F}_q$. Code $C$ is a subset of $[q]^n$ and $|C|= q^r$, where $[q] = \{0, 1, \ldots, q-1 \}$. Each element in $C$ is called a codeword, which is a vector of size $n \times 1$. Each entry of a codeword belongs to $\mathbb{F}_q$. $\Delta$ is the minimum distance, that is the minimum number of positions in which any two codewords of $C$ differ. $C$ is called an MDS code if $\Delta = n - r + 1$. The following lemma is useful to enumerate how many codewords in $C$ if we fix a coordinate position.

\begin{lemma}[Section 3~\cite{silverman1960metrization}]
\label{lem:countSimilarity}

Let $C$ be an $[n, r, n - r + 1]_q$-MDS code over the alphabet $\mathbb{F}_q$. Fix a position $i_0$ and let $\alpha$ be an element of $\mathbb{F}_q$. Then there are exactly $q^{r-1}$ codewords with $\alpha$ in the position $i_0$.
\end{lemma}

In this paper, we choose $C$ as a Reed-Solomon code~\cite{reed1960polynomial}. Set $N = q^r$. Let $\mathcal{C} = (c_{ij})$ be the $n \times N$ matrix constructed by putting all codewords of $C$ as its columns. Because of the construction of Reed-Solomon code, $\mathcal{C}$ is nonrandomly constructed.

We next introduce concatenation technique. . LLet $\mathcal{I}$ be an $q \times q$ identity matrix. The concatenation code between $\mathcal{C}$ and $\mathcal{I}$ in which the $h \times N$ resulting matrix $\mathcal{G}$, denoted as $\mathcal{G} = \mathcal{C} \circ \mathcal{I}$, is defined as:

\begin{equation}
\mathcal{G} = \mathcal{C} \circ \mathcal{I} = \begin{bmatrix}
\mathcal{C}_1 \circ \mathcal{I} & \ldots \mathcal{C}_N \circ \mathcal{I}
\end{bmatrix} =
\begin{bmatrix}
\mathcal{I}_{c_{11} + 1} & \mathcal{I}_{c_{12} + 1} & \ldots & \mathcal{I}_{c_{1N} + 1} \\
\mathcal{I}_{c_{21} + 1} & \mathcal{I}_{c_{22} + 1} & \ldots & \mathcal{I}_{c_{2N} + 1} \\
\vdots & \vdots & \ddots & \vdots \\
\mathcal{I}_{c_{n1} + 1} & \mathcal{I}_{c_{n2} + 1} & \ldots & \mathcal{I}_{c_{nN} + 1}
\end{bmatrix}
\end{equation}
\noindent
where $h = nq$ and $\mathcal{C}_j$ is the $j$th column of $\mathcal{C}$ for $j = 1, \ldots, N$. In short, the concatenation technique maps an element $\beta$ in $\mathbb{F}_q$ to a column $\mathcal{I}_{\beta + 1}$ of matrix $\mathcal{I}$. From this construction, the Hamming weight of each column in $\mathcal{G}$ is $n$. In term of matrices, an example of concatenated code is presented below in which $q = 8, n = 7, r = 3$. Then $N = q^r = 512$. Since $\mathcal{C}$ is a $7 \times 512$ matrix, which is rather large, we only present a concatenation of a column of $\mathcal{C}$ and $\mathcal{I}$ as follows:
\begin{eqnarray}
\mathcal{C}_2 &=& \begin{bmatrix} 1 & 0 & 0 & 6 & 1 & 6 & 7 \end{bmatrix}^T; \\
\mathcal{C}_2 \circ \mathcal{I} &=& \begin{bmatrix} \mathcal{I}_{2, *} & \mathcal{I}_{1, *} & \mathcal{I}_{1, *} & \mathcal{I}_{7, *} & \mathcal{I}_{2, *} & \mathcal{I}_{7, *} & \mathcal{I}_{8, *} \end{bmatrix}^T;
 \label{eqn:TensorMat}
\end{eqnarray}

It is known that for any $d = \lfloor \frac{n - 1}{r - 1} \rfloor$, $\mathcal{G}$ is a $h \times N$ $d$-disjunct matrix, where $h < d^2 \log_2^2{N}$~\cite{kautz1964nonrandom}. Moreover, the matrix $\mathcal{G}$ is nonrandomly constructed because $\mathcal{C}$ and $\mathcal{I}$ are nonrandomly constructed.

Our task now is to prove that the Hamming weight of each row in $\mathcal{G}$ is equal to $q^{r-1}$. Indeed, from Lemma~\ref{lem:countSimilarity}, any element $\alpha$ of $\mathbb{F}_q$ appears exactly $q^{r-1}$ times in any row of $\mathcal{C}$. Because $\mathcal{C}$ is concatenated with an $q \times q$ identity matrix $\mathcal{I}$, the Hamming weight of every row in $\mathcal{G}$ would be equal to $q^{r-1} < \frac{q^r}{2} = N/2$.

\end{document}